\documentclass[twocolumn]{autart}
\usepackage{amsmath,amssymb,amsfonts}
\usepackage{graphicx}
\usepackage{epstopdf}
\usepackage{cite}
\usepackage{algorithmic}
\usepackage{textcomp}
\usepackage{xcolor}
\usepackage{subfig}
\usepackage{bm}
\usepackage{cases}
\usepackage{tikz}
\usepackage{setspace}
\usetikzlibrary{arrows,shapes,chains}
\newtheorem{theorem}{Theorem}

\newtheorem{assumption}{Assumption}

\newtheorem{problem}{Problem}

\newtheorem{proof}{Proof}
\begin{document}

\begin{frontmatter}

\title{When Distributed Formation Control Is Feasible under Hard Constraints on Energy and Time?  \thanksref{footnoteinfo}}
\thanks[footnoteinfo]{This work was supported in part by the National Natural Science Foundation of China under Grants 61973064 and 61973061, and in part by the Hebei Natural Science Foundation for Distinguished Young Scholars under Grant F2019501043.  Corresponding author: Fei Chen (fei.chen@ieee.org).}

\author[NEU,NEUQ]{Chunxiang Jia},
\author[NEU,NEUQ]{Fei Chen},
\author[NEU,NEUQ]{Linying Xiang},
\author[XMU]{Weiyao Lan}, 
\author[CityU]{Gang Feng}
\address[NEU]{State Key Laboratory of Synthetical Automation for Process Industries, Northeastern University, Shenyang, 110004, China}
\address[NEUQ]{School of Control Engineering, Northeastern University at Qinhuangdao, Qinhuangdao, 066004, China}
\address[XMU]{Department of Automation, Xiamen University, Xiamen, 361005, China}
\address[CityU]{Department of Biomedical Engineering, City University of Hong Kong, Kowloon, Hong Kong SAR, China}


\begin{abstract}
This paper studies distributed optimal formation control with hard constraints on energy levels and termination time, in which the formation error is to be minimized jointly with the energy cost. The main contributions include a globally optimal distributed formation control law and a comprehensive analysis of the resulting closed-loop system under those hard constraints. It is revealed that the
energy levels, the task termination time, the steady-state error tolerance, as well as the network topology impose inherent limitations in achieving the formation control mission. Most notably, the lower bounds on the achievable termination time and the required minimum energy levels are derived, which are given in terms of the initial formation error, the steady-state error tolerance, and the largest eigenvalue of the Laplacian matrix. These lower bounds can be employed to assert whether an energy and time constrained formation task is achievable and how to accomplish such a task. Furthermore, the monotonicity of those lower bounds in relation to the control parameters is revealed. A simulation example is finally given to illustrate the obtained results.
\end{abstract}

\end{frontmatter}

\section{Introduction}
\label{sect:intr}
This paper is concerned with energy and time constraints and performance tradeoff issues one frequently encounters in distributed formation control of multi-agent systems. A fundamental problem under investigation is how energy level, mission termination time, and steady-state error tolerance may inherently impact on the achievable performance of formation control, and how such impacts may be quantified analytically. Formation control problems have been widely studied in the recent literature (see, e.g., \cite{oh2015survey,su2009flocking,chen2017connection,beard2001coordination,balch1998behavior,lin2005necessary} and the references therein). However, only a rather limited number of works have considered energy constraints \cite{weimerskirch2001energy,derenick2011energy,papakostas2018energy,sardellitti2011optimal}, though the issue is of significant importance for agents with limited energy supplied by on-board batteries.

The energy and time constraints impose severe limitations on distributed cooperative control design and have motivated several existing works involving various cooperative tasks \cite{babazadeh2018cooperative,babazadeh2018anoptimal,zhang2018consensus,moarref2014optimal,mei2015distributed,xiang2019advances}, wherein the energy cost is defined as an integral of the square of the input, and is to be minimized, together, with certain control error functions. Other relevant attempts have been pursued by researchers to reduce redundant communication to decrease the energy cost \cite{demirel2017trade,varma2019energy}. In addition, it has been recognized that the resistance caused by velocity mismatches may also contribute to the energy expenditure, which cannot be ignored for systems with relatively high velocities \cite{niu2017numerical,chu2014numerical}.

The LQR-based method is just one case of many efforts which seek to limit the energy consumption. It is noted that a direct application of the LQR-based method to multi-agent systems will generically require an all-to-all network topology (see, e.g., \cite{cao2009optimal,di2012rendezvous}). That is, there is a dilemma between distributed control and LQR-based optimal control. Very recently, a network approximation approach  is developed in \cite{chen2019minimum} by introducing a ``minimal'' distribution cost in the LQR function, which guarantees that the resulting control law is optimal in the global sense.

The present paper continues the aforementioned development in the study of energy-aware formation control of multi-agent systems. The main contributions are three-fold. Firstly, a distributed formation control law is derived which is globally optimal with respect to a cost pertinent to energy and control error of the multi-agent system under the LQR framework. To the best of the authors' knowledge, the proposed algorithm is the first formation control algorithm that is concurrently distributed and optimal while satisfying the hard constraints on energy expenditure and convergence time. Secondly, the conditions on the feasibility of the formation control problem are derived analytically, which depends upon the initial energy level, the formation termination time, the steady-state error tolerance, the network topology, as well as the control parameters. Thirdly, monotonicity properties of the achievable termination time and the required minimum initial energy  with respect to the control parameters are further revealed, which provides some design guidelines in achieving formation control missions under time and energy constraints. A preliminary version of the results discussed here has appeared in \cite{jia2020distributed}. With respect to \cite{jia2020distributed}, the current version provides a comprehensive analysis on the monotonicity properties of the PARE solution, the termination time, as well as the energy expenditure. Moreover, numerical examples are also provided to illustrate the validity of the proposed results. 



The rest of this paper is organized as follows. In Section~\ref{sect:prel}, preliminaries are presented and the problem is formulated. Section~\ref{sect:contr} is devoted to the development of the optimal distributed control algorithm and its analysis. Section~\ref{sec:IV} discusses the monotonicity properties of the achievable termination time and the required minimum energy with respect to the control parameters. Simulation results are presented in Section~\ref{sect:sim}. Finally, Section~\ref{sect:concl} concludes the paper.

\section{Preliminaries and problem statement}
\label{sect:prel}

\subsection{Notation}
Let $\mathbb{R}$ denote the set of real numbers, $\mathbb{R}^+$ the set of positive real numbers, $\mathbb{R}^n$ the set of $n$-dimensional real vectors, and $\mathbb{R}^{n\times n}$ the set of $n\times n$ real matrices.  Let $I_n\in \mathbb{R}^{n\times n}$ be the $n$-dimensional identity matrix, $\mathbf{0}_n\in \mathbb{R}^n$ the vector with all zeros, and $\mathbf{1}_n\in \mathbb{R}^n$ the vector with all ones. The subscripts of $I_n$, $\mathbf{0}_n$, and $\mathbf{1}_n$ might be dropped if no confusion arises from the context. The superscript $T$ denotes the transpose of a matrix or a vector. The set of the eigenvalues of $A$ is denoted by $\mathrm{spec}(A)$. The Euclidean norm is given by $\Vert\cdot\Vert$. For two matrices $A\in \mathbb{R}^{m\times n}$ and $B\in \mathbb{R}^{p\times q}$, their Kronecker product is denoted by
\begin{align*}
A\otimes B=\left[
             \begin{array}{ccc}
               a_{11}B & \cdots & a_{1n}B \\
               \vdots & \ddots & \vdots \\
               a_{m1}B & \cdots & a_{mn}B \\
             \end{array}
           \right].
\end{align*}
The abbreviation ``iff'' means ``if and only if''.

\subsection{Graph theory}
The information exchange among the agents is described by a graph $\mathcal{G}=(\mathcal{V},\mathcal{E})$, where $\mathcal{V}=\{\nu_1,\dots,\nu_N\}$ is the set of nodes and $\mathcal{E} \subseteq \mathcal{V}\times \mathcal{V}$ is the set of edges. In this paper, the graph $\mathcal{G}$ is assumed to be undirected. The adjacency matrix $\mathcal{A}=[a_{ij}]\in \mathbb{R}^{N\times N}$ of $\mathcal{G}$ is defined as: $a_{ij}=1$ if $(i,j)\in \mathcal{E}$, and $a_{ij}=0$ otherwise. The degree matrix is then given by $\mathcal{D}=\mathrm{diag}([d_1,\dots,d_N])$, where $d_i=\sum_{j=1}^{N}a_{ij}$. A path from node $\nu_i$ to node $\nu_j$ is a sequence of nodes $\nu_i,\dots,\nu_j$, such that each two consecutive nodes in the sequence is connected by an edge. An undirected graph is connected if for any two vertices in $\mathcal{V}$, there always exists a path connecting them. Throughout the paper, the following assumption is made.
\begin{assumption}\label{assump1}
Graph $\mathcal{G}$ is undirected and connected.
\end{assumption}
The Laplacian matrix of the undirected graph $\mathcal{G}$ is given by $\mathcal{L}=\mathcal{D}-\mathcal{A}\in \mathbb{R}^{N\times N}$, which is known to be symmetric and positive semi-definite. It has a zero eigenvalue whose normalized eigenvector is $\frac{1}{\sqrt{N}}\mathbf{1}_N$, where $\textbf{1}_N\in \mathbb{R}^N$ is the vector with all ones. The $N$ real eigenvalues of $\mathcal{L}$ can be ordered as $0=\lambda_1\leq\lambda_2\leq\dots\leq\lambda_N$. Let $\mathcal{W}=[w_1,\dots, w_N]^T$ be the matrix comprising orthonormal eigenvectors of $\mathcal{L}$. The Laplacian matrix $\mathcal{L}$ can be diagonalized as follows:
\begin{align}\label{1}
\mathcal{L}=\mathcal{W}^T \mathcal{J} \mathcal{W},
\end{align}
where $\mathcal{J}=\mathrm{diag}([\lambda_1,\dots,\lambda_N])$.

\subsection{Problem statement}
Consider a multi-agent system consisting of $N$ agents moving in the $n$-dimensional space. Each agent is governed by the following equations:
\begin{align}
\dot{p}_i(t)&=v_i(t), \qquad \dot{v}_i(t)=u_{i}(t), \label{eq:dyn}\\
\dot{E}_i(t)&=-u_i^T(t)u_i(t)-\frac{\beta}{2}\sum_{i=1}^Na_{ij}\Vert v_i(t)-v_j(t)\Vert^2,\label{eq:enDyn}\\
p_i(0)&=p_i^0, \quad v_i(0)=v_i^0, \quad E_i(0)=E_i^0,\quad i=1,\dots,N, \nonumber
\end{align}
where $p_i(t)\in \mathbb{R}^n$, $v_i(t)\in \mathbb{R}^n$, $u_i(t)\in \mathbb{R}^n$, and $E_i(t) \in \mathbb{R}$ denote, respectively, the position, velocity, input, and energy level of agent $i$, and $p_i^0\in \mathbb{R}^n$, $v_i^0\in \mathbb{R}^n$, and $E_i^0\in \mathbb{R}$ are their initial values. Equation~\eqref{eq:dyn} describes the double-integrator dynamics of the agents, while Eq.~\eqref{eq:enDyn} delineates how the energy level of the agents changes. The first term of \eqref{eq:enDyn} represents the energy expenditure caused by the control input, while the second term represents the energy expenditure due to the resistance of velocity mismatch, where $\beta$ is a positive constant. Let
\begin{align}
  J_E^i(t)&=\int_0^t -\dot{E}_i(\tau)d\tau \nonumber
\end{align}
be the energy consumed by agent $i$ till time $t$. The energy cost of the multi-agent system is given by
\begin{align}\label{5}
J_E(t)&=\sum_{i=1}^N\int_0^t -\dot{E}_i(\tau)d\tau\nonumber\\
&=\int_0^t\{u^T(\tau)u(\tau)+\beta v^T(\tau)(\mathcal{L}\otimes I_n)v(\tau)\}d\tau,
\end{align}
where $u(\tau)=[u_1^T(\tau),\dots,u_N^T(\tau)]^T\in \mathbb{R}^{Nn}$ and $v(\tau)=[v_1^T(\tau),\dots,v_N^T(\tau)]^T\in\mathbb{R}^{Nn}$.
For notational convenience, $J_E(\infty)$ will be simplified as $J_E$ in the rest of the paper.

Define $x_i(t)=[p_i^T(t)\quad v_i^T(t)]^T\in \mathbb{R}^{2n}$. Equation~(\ref{eq:dyn}) can be written compactly as
\begin{align}\label{6}
\dot{x}_i(t)&=Ax_i(t)+Bu_i(t),
\end{align}
where
$A=\left[\begin{array}{cc}
   0 & 1 \\
   0 & 0 \\
\end{array}\right]\otimes I_n$ and
$
B=\left[\begin{array}{c}
   0 \\
   1 \\
\end{array}\right]\otimes I_n$.
Let $x^d=[(p^d)^T\quad (v^d)^T]^T\in \mathbb{R}^{2Nn}$ represents the desired state with $p^d=[(p_1^d)^T,\dots,(p_N^d)^T]^T\in \mathbb{R}^{Nn}$ and $v^d=[(v_1^d)^T,\dots,(v_N^d)^T]^T\in \mathbb{R}^{Nn}$ denoting, respectively, the desired position and velocity. To guarantee the tracking result, it is necessary that all agents have the same desired velocity. Particularly, for notational convenience, it is assumed that $v^d=\mathbf{0}$. Accordingly, the energy cost function (\ref{5}) can be rewritten in terms of $u(t)$ and $x(t)$ as
\begin{align}\label{7}
J_E=&\int_0^\infty\{u^T(t)u(t)+\beta[x(t)-x^d]^T(\mathcal{L}\otimes Q)\nonumber\\
&\times[x(t)-x^d]\}dt,
\end{align}
where $Q=\mathrm{diag}([0\quad1])\otimes I_n$, $x(t)=[x_1^T(t),\dots,x_N^T(t)]\in \mathbb{R}^{2Nn}$. Let $d_{ij}=[(p_{ij}^d)^T\quad (v_{ij}^d)^T]^T$ denote the prespecified relative state between agent $i$ and $j$, i.e., $p_{ij}^d=p_i^d-p_j^d$ and $v_{ij}^d=v_i^d-v_j^d=0$. Let $T$ be the termination time of the formation task, and $\varepsilon \in \mathbb{R}^+$ be the parameter of the steady-state error tolerance. The following problem is investigated in the paper.
\begin{problem}\label{prob1}
Design a distributed control input $u_i(t)$ for the system \eqref{6}, based on local information, such that for some $t_f \in \mathbb{R}^+$,
\begin{align}
& \quad \limsup_{t\rightarrow t_f}\Vert x_i(t)-x_j(t)-d_{ij}\Vert\leq\varepsilon \nonumber  \\
\mathrm{s. t.} & \quad t_f\leq T, \quad J_E^i(T)<E_i^0.
\end{align} \nonumber
\end{problem}
It is worth pointing out that $t_f\leq T$ and $J_E^i(T)<E_i^0$ are two ``hard'' constraints on the formation task. If $t_f>T$, the formation task fails to be achieved since it is not accomplished in a timely manner. On the other hand, $J_E^i(T) \geq E_i^0$ means that the energy is exhausted before the mission is completed.

\section{Distributed optimal energy-aware formation control}
\label{sect:contr}
This section is devoted to the development of an energy-aware distributed formation control algorithm by employing solely local information.
To this aim, define the performance measure
\begin{align}
J=J_E+J_x^f+J_x^{NA}, \nonumber
\end{align}
where the energy cost $J_E$ is defined in \eqref{7}, and
\begin{align*}
J_x^f&=\alpha\int_0^\infty [x(t)-x^d]^T(\mathcal{L}\otimes I_{2n})[x(t)-x^d]dt\\
&=\frac{\alpha}{2}\int_0^\infty\sum_{i=1}^Na_{ij}\|x_i(t)-x_j(t)-d_{ij}\|^2dt,\\
J_x^{NA}&=\alpha\int_0^\infty [x(t)-x^d]^T[M\otimes S][x(t)-x^d]dt.
\end{align*}
Here, $\alpha>0$ is a tradeoff parameter, 
$M=\alpha(\mathcal{L}^2-\sigma \mathcal{L})$ with $0<\sigma<\lambda_2$, and $S\geq0$ is a positive semi-definite matrix to be designed. The formation cost term $J_x^f$ represents the accumulated formation error, and ensures that the formation is reached asymptotically. It has been recognized that for a multi-agent system, the LQR-based optimal control law only exists under an all-to-all network topology \cite{cao2009optimal}. To circumvent the difficulty, the distribution cost term $J_x^{NA}$ is introduced to warrant that the optimal distributed control law exists for a generic connected network topology \cite{chen2019minimum}. The main results of this section are given as follows.
\begin{theorem}\label{the1}
Let
\begin{align}
P=\frac{1}{\sqrt{\sigma\alpha}}\left[
    \begin{array}{cc}
      \sqrt{\sigma\alpha+\beta\sigma+2\sqrt{\sigma\alpha}} & 1 \\
      1 & \sqrt{1+\frac{\beta}{\alpha}+\frac{2}{\sqrt{\sigma\alpha}}} \\
    \end{array}
  \right]\otimes I_n, \nonumber
\end{align}
where $0<\sigma<\lambda_2$ with $\lambda_2$ being the second smallest eigenvalue of the Laplacian matrix $\mathcal{L}$. If $M=\alpha(\mathcal{L}^2-\sigma \mathcal{L})$, $S=PBB^TP$, and Assumption \ref{assump1} holds, then
\begin{enumerate}
  \item the optimal distributed control input of (\ref{6}) that minimizes $J$ is given by
\begin{align}\label{11}
u_i^*(t)=-\alpha\sum_{j=1}^Na_{ij}B^TP[x_i^*(t)-x_j^*(t)-d_{ij}],
\end{align}
where $x_i^*(t)$ is the state under the optimal input $u_i^*(t)$ at time $t$;
  \item for given initial energy $E(0)=[E_1(0),\dots,E_N(0)]^T\in \mathbb{R}^N$, termination time $T>0$, and steady-state error tolerance $\varepsilon>0$, if the following inequalities hold
       \begin{numcases}{}
          T\geq \lambda_{\min}(P)\ln{\frac{V(x(0))}{\lambda_{\min}(P)(N-1)\varepsilon^2}},\label{12}\\
          E_i(0)\geq \frac{V_\mathcal{L}(0)}{2}\bigg[\lambda_N\bigg(\alpha+\frac{1}{\sigma}\bigg)\bigg(\alpha+\beta+2\sqrt{\frac{\alpha}{\sigma}}\bigg)\nonumber\\
          +\beta\bigg]\sqrt{1+\frac{\beta}{\alpha}+\frac{2}{\sqrt{\alpha\sigma}}}\bigg(1-e^{-\lambda_N\sqrt{\frac{\alpha}{\sigma}(1+\frac{\beta}{\alpha}+\frac{2}{\sqrt{\alpha\sigma}})}T}\bigg),\nonumber\\  \qquad i\in\{1,\dots,N\}, \label{13}
       \end{numcases}
where
\begin{align*}
&\lambda_{\min}(P)=\frac{1}{2\sqrt{\sigma\alpha}}\bigg(\bigg(1+\sqrt{\sigma\alpha}\bigg)\sqrt{1+\frac{\beta}{\alpha}+\frac{2}{\sqrt{\sigma\alpha}}}\\
&\quad-\sqrt{\bigg(1+\sigma\alpha-2\sqrt{\sigma\alpha}\bigg)\bigg(1+\frac{\beta}{\alpha}+\frac{2}{\sqrt{\sigma\alpha}}\bigg)+4}\bigg),\\
&V(x(0))=[x(0)-x^d]^T\left[\left(I_N-\frac{1}{N}\textbf{11}^T\right)\otimes P\right]\nonumber \\
& \hspace*{40pt} \times [x(0)-x^d],\\
&V_\mathcal{L}(0)=[x(0)-x^d]^T\left(\mathcal{L}\otimes I_{2n}\right)[x(0)-x^d],
\end{align*}
then Problem \ref{prob1} is solved under the distributed optimal control algorithm \eqref{11}.
\end{enumerate}
\end{theorem}
\begin{proof}
1) Define
\begin{align}
&\quad J(t_f,x(t_f))\nonumber\\
&=\int_0^{t_f}\{u^T(t)u(t)+\beta [x(t)-x^d]^T(\mathcal{L}\otimes Q)[x(t)-x^d]\}dt\nonumber\\
&\quad +\alpha \bigg\{\int_0^{t_f}[x(t)-x^d]^T(\mathcal{L}\otimes I_{2n})[x(t)-x^d]dt\nonumber\\
&\quad +\int_0^{t_f}[x(t)-x^d]^T[M\otimes S(t)][x(t)-x^d]dt\nonumber\\
&\quad +[x(t_f)-x^d]^T(\mathcal{L}\otimes I_{2n})[x(t_f)-x^d]\bigg\},\nonumber
\end{align}
where $S(t)\in \mathbb{R}^{2n\times 2n}$ is a time-varying positive semi-definite matrix, and $t_f$ is the actual convergence time defined in Problem \ref{prob1}. Let $\tilde{x}_i(t)$ and $\tilde{u}_i(t)$ denote, respectively, the $i$th component of $\tilde{x}(t)\triangleq(\mathcal{W} \otimes I_{2n})[x(t)-x^d]$ and $\tilde{u}(t)\triangleq(\mathcal{W} \otimes I_{2n})u(t)$, where $\mathcal{W}$ is defined in \eqref{1}. The multi-agent system (\ref{6}) can be written equivalently as
\begin{align}\label{15}
\dot{\tilde{x}}_i(t)=A\tilde{x}_i(t)+B\tilde{u}_i(t),\quad i=1,\dots,N,
\end{align}
where $Ax_i^d=0$, $i=1,\dots,N,$ is used. Due to Assumption \ref{assump1}, $J(t_f,x(t_f))$ can be written equivalently as $J(t_f,x(t_f))=\sum_{i=1}^NJ_i(t_f,\tilde{x}_i(t_f)),$

where
\begin{align}\label{16}
J_1(t_f,\tilde{x}_1(t_f))&=\int_0^{t_f}\tilde{u}_1^T(t)\tilde{u}_1(t)dt,\nonumber\\
J_i(t_f,\tilde{x}_i(t_f))&=\int_0^{t_f}\{\tilde{u}_i^T(t)\tilde{u}_i(t)+ \lambda_i\beta \tilde{x}_i^T(t)Q\tilde{x}_i(t)\}dt\nonumber\\
& +\alpha \bigg\{\int_0^{t_f}\tilde{x}_i^T(t)[\lambda_iI_{2n}+m_iS(t)]\tilde{x}_i(t)dt\nonumber\\
&+\lambda_i\tilde{x}_i^T(t_f)\tilde{x}_i(t_f)\bigg\},\quad i=2,\dots,N
\end{align}
with $m_i\triangleq \alpha(\lambda_i^2-\sigma\lambda_i)$. It is straightforward to obtain that $\tilde{u}_1^*\equiv0$.

Next, the optimal input $\tilde{u}_i^*$ is derived for $i=2,\dots,N$. Let $\tilde{x}_i^*(t)$ denote the state of (\ref{15}) under the optimal input $\tilde{u}_i^*(t)$, i.e.,
\begin{align}\label{17}
\dot{\tilde{x}}_i^*(t)=A\tilde{x}_i^*(t)+B\tilde{u}_i^*(t)
\end{align}
with the initial condition $\tilde{x}_i^*(0)=\tilde{x}_i^0$, where $\tilde{x}_i^0$ is the $i$th component of $\tilde{x}^0=(\mathcal{W}\otimes I_{2n})(x^0-x^d)$. Consider a new input vector
\begin{align}\label{18}
\tilde{u}_i(t)=\tilde{u}_i^*(t)+\epsilon \hat{u}_i(t)
\end{align}
for (\ref{15}), where $\hat{u}_i(t)$ is an arbitrary function of time, and $\epsilon\in \mathbb{R}$ is an arbitrary number. Due to the variation of the input vector, the state of the system (\ref{15}) will change from $\tilde{x}_i^*(t)$ to
\begin{align}\label{19}
\tilde{x}_i(t)=\tilde{x}_i^*(t)+\epsilon\hat{x}_i(t), \quad 0\leq t\leq t_f,
\end{align}
where $\hat{x}_i(t)$ is some function of time. Substitution of (\ref{18}) and (\ref{19}) into (\ref{15}) yields
\begin{align}\label{20}
\dot{\tilde{x}}_i^*(t)+\epsilon\dot{\hat{x}}_i(t)=&A[\tilde{x}_i^*(t)+\epsilon\hat{x}_i(t)]+B[\tilde{u}_i^*(t)+\epsilon\hat{u}_i(t)].
\end{align}
Substraction of \eqref{17} from \eqref{20} and cancelation of $\epsilon$ lead to
\begin{align}\label{21}
\dot{\hat{x}}_i(t)=A\hat{x}_i(t)+B\hat{u}_i(t)
\end{align}
with the initial condition $\hat{x}_i(0)=0$. The solution of (\ref{21}) is
\begin{align}\label{22}
\hat{x}_i(t)=\int_0^te^{A(t-\tau)}B\hat{u}_i(\tau)d\tau.
\end{align}
Using (\ref{18}) and (\ref{19}), Equation~(\ref{16}) can be rewritten as a function related to $\epsilon$, denoted by $J_i(t_f,\tilde{x}_i(t_f),\epsilon)$. Since $\tilde{u}_i^*(t)$ is the control input that minimizes $J_i(t_f,\tilde{x}_i(t_f),\epsilon)$, $J_i(t_f,\tilde{x}_i(t_f),\epsilon)$ must have a minimum at $\epsilon=0$, which implies that the first derivative of $J_i(t_f,\tilde{x}_i(t_f),\epsilon)$ with respect to $\epsilon$ should be zero at $\epsilon=0$. It thus follows that
\begin{align}\label{23}
&\int_0^{t_f}\{\hat{u}_i^T(t)\tilde{u}_i^*(t)+\hat{x}_i^T(t)[\lambda_i\beta Q+\lambda_i\alpha I_{2n}+m_i\alpha S(t)]\nonumber\\
&\times\tilde{x}_i^*(t)\}dt+\lambda_i\alpha\hat{x}_i^T(t_f)\tilde{x}_i^*(t_f)=0.
\end{align}
Substitution of (\ref{22}) into (\ref{23}) together with some rearrangements leads to
\begin{align}\label{24}
&\int_0^{t_f}\hat{u}_i^T(t)\bigg\{\tilde{u}_i^*(t)+B^T\int_t^{t_f}e^{A^T(\tau-t)}[\lambda_i\beta Q+\lambda_i\alpha I_{2n}\nonumber\\
&+m_i\alpha S(\tau)]\tilde{x}_i^*(\tau)d\tau+\lambda_i\alpha B^Te^{A^T(t_f-t)}\tilde{x}_i^*(t_f)\bigg\}dt=0.
\end{align}
Let
\begin{align}\label{25}
p_i(t)\triangleq&\int_t^{t_f}e^{A^T(\tau-t)}[\lambda_i\frac{\beta}{\alpha}Q+\lambda_i I_{2n}+m_i S(\tau)]\tilde{x}_i^*(\tau)d\tau\nonumber\\
&+\lambda_i e^{A^T(t_f-t)}\tilde{x}_i^*(t_f).
\end{align}
Equation~(\ref{24}) can be written compactly as
\begin{align}\label{26}
\int_0^{t_f}\hat{u}_i^T(t)\{\tilde{u}_i^*(t)+\alpha B^Tp_i(t)\}dt=0.
\end{align}
Since (\ref{26}) holds for all possible $\hat{u}_i(t)$, it follows that
\begin{align}\label{27}
\tilde{u}_i^*(t)=-\alpha B^Tp_i(t).
\end{align}
Therefore, the problem of finding the optimal input $\tilde{u}_i^*(t)$ is transformed into the problem of finding the solution of $p_i(t)$ that satisfies (\ref{25}). Similar to the process of obtaining $p_i(t)$ in \cite{chen2019minimum}, it can be shown that
\begin{align}\label{28}
p_i(t)=\lambda_iP(t)\tilde{x}_i^*(t),
\end{align}
where $P(t)$ is the solution to the following parametric differential Riccati equation (PDRE)
\begin{equation}\label{29}
\dot{P}(t)+ I_{2n}+\frac{\beta}{\alpha}Q+A^TP(t)+P(t)A-\sigma\alpha P(t)BB^TP(t)=0,
\end{equation}
\begin{align}
P(t_f)= I_{2n}, \nonumber
\end{align}
where $0<\sigma<\lambda_2$. Substitution of (\ref{28}) into (\ref{27}) yields
\begin{align*}
\tilde{u}_i^*(t)=-\alpha\lambda_iB^TP(t)\tilde{x}_i^*(t),
\end{align*}
or equivalently
\begin{align}
\tilde{u}^*(t)=-\alpha[\mathcal{J}\otimes B^TP(t)]\tilde{x}^*(t), \nonumber
\end{align}
which can be further written as
\begin{align}
u^*(t)=-\alpha[\mathcal{L}\otimes B^TP(t)][x^*(t)-x^d]. \nonumber
\end{align}\\
Let $P$ be the solution to the following parametric algebraic Riccati equation (PARE)
\begin{align}\label{33}
I_{2n}+\frac{\beta}{\alpha} Q+A^TP+PA-\sigma\alpha PBB^TP=0.
\end{align}
Since $(A,B)$ is stabilizable and $(A,I_{2n})$ is detectable, the solution to (\ref{29}) converges to that of (\ref{33}) as $t_f\rightarrow \infty$. This leads to the optimal control input in the infinite-horizon case,
\begin{align}\label{34}
\tilde{u}_i^*(t)=-\alpha\lambda_iB^TP\tilde{x}_i^*(t),
\end{align}
or equivalently
\begin{align}
u^*(t)=-\alpha(\mathcal{L}\otimes B^TP)[x^*(t)-x^d]. \nonumber
\end{align}
Substituting $A=\left[\begin{array}{cc}
   0 & 1 \\
   0 & 0 \\
\end{array}\right]\otimes I_n$ and
$
B=\left[\begin{array}{c}
   0 \\
   1 \\
\end{array}\right]\otimes I_n$ into \eqref{33}, the solution $P$ to \eqref{33} is given by
\begin{align}\label{36}
P=\frac{1}{\sqrt{\sigma\alpha}}\left[
    \begin{array}{cc}
      \sqrt{\sigma\alpha+\beta\sigma+2\sqrt{\sigma\alpha}} & 1 \\
      1 & \sqrt{1+\frac{\beta}{\alpha}+\frac{2}{\sqrt{\sigma\alpha}}} \\
    \end{array}
  \right]\otimes I_n.
\end{align}
The proof of the first part is thus completed.

2) Substituting the optimal control law (\ref{11}) into the system \eqref{6} yields the following closed-loop system:
\begin{align}\label{37}
\dot{x}^*(t)=(I_N\otimes A)x^*(t)-\alpha(\mathcal{L}\otimes BB^TP)[x^*(t)-x^d].
\end{align}
Define $V(x)=[x(t)-x^d]^T\left[(I_N-\frac{1}{N}\textbf{11}^T)\otimes P\right][x(t)-x^d],$
where $P$ is given by (\ref{36}). Note that $V(x)=0$ iff the formation is reached.
It follows from (\ref{37}) that
\begin{align}\label{38}
&\quad\dot{V}(x^*)\nonumber\\
&=2[x^*(t)-x^d]^T\bigg[\bigg(I_N-\frac{1}{N}\textbf{11}^T\bigg)\otimes PA\bigg][x^*(t)-x^d]\nonumber\\
&-2[x^*(t)-x^d]^T[\mathcal{L}\otimes \alpha PBB^TP][x^*(t)-x^d].
\end{align}
The first term in (\ref{38}) can be written as
\begin{align}\label{39}
&\quad 2[x^*(t)-x^d]^T\bigg[\bigg(I_N-\frac{1}{N}\textbf{11}^T\bigg)\otimes PA\bigg][x^*(t)-x^d]\nonumber\\
&=2[x^*(t)-x^d]^T(\mathcal{W}^T\Upsilon \mathcal{W}\otimes PA)[x^*(t)-x^d]\nonumber\\
&=2\sum_{i=2}^N[\tilde{x}_i^*(t)]^TPA\tilde{x}_i^*(t),
\end{align}
where $\Upsilon=\mathrm{diag}([0,1,\dots,1])\in \mathbb{R}^{N\times N}$. Similarly, the second term can be rewritten as
\begin{align}\label{40}
&\quad2[x^*(t)-x^d]^T[\mathcal{L}\otimes\alpha PBB^TP][x^*(t)-x^d]\nonumber\\
&=2\sum_{i=2}^N\alpha\lambda_i[\tilde{x}_i^*(t)]^TPBB^TP\tilde{x}_i^*(t).
\end{align}
Substituting (\ref{40}) and (\ref{39}) into (\ref{38}) yields
\begin{align}\label{41}
\dot{V}(x^*)&=\sum_{i=2}^N[\tilde{x}^*_i(t)]^T[(A-\lambda_i\alpha BB^TP)^TP \nonumber\\
&\quad +P(A-\lambda_i\alpha BB^TP)]\tilde{x}_i^*(t)\nonumber\\
&=-\sum_{i=2}^N[\tilde{x}^*_i(t)]^T\bigg\{I_{2n}+\frac{\beta}{\alpha} Q+\alpha[\sigma+2(\lambda_i-\sigma)] \nonumber\\
 &\quad\times PBB^TP\bigg\}\tilde{x}^*_i(t)\nonumber\\
&\leq-\sum_{i=2}^N[\tilde{x}_i^*(t)]^T\tilde{x}_i^*(t),
\end{align}
where the second equality is due to $(A-\lambda_i\alpha BB^TP)^TP+P(A-\lambda_i\alpha BB^TP)=A^TP-\sigma\alpha PBB^TP+PA-\sigma\alpha PBB^TP-2\alpha(\lambda_i-\sigma)PBB^TP
=-\bigg( I_{2n}+\frac{\beta}{\alpha} Q+\alpha[\sigma+2(\lambda_i-\sigma)]PBB^TP\bigg).$
Additionally,
\begin{align}\label{42}
V(x^*)&=[x^*(t)-x^d]^T\bigg[(I_N-\frac{1}{N}\textbf{11}^T)\otimes P\bigg][x^*(t)-x^d]\nonumber\\
&\geq\lambda_{\min}(P)\sum_{i=2}^N[\tilde{x}_i^*(t)]^T\tilde{x}_i^*(t).
\end{align}
It follows from \eqref{41} and \eqref{42} that $\frac{\dot{V}(x^*)}{V(x^*)}\leq-\frac{1}{\lambda_{\min}(P)},$
which gives
\begin{align}\label{44}
V(x^*(t))\leq e^{-\frac{1}{\lambda_{\min}(P)}t}V(x(0)).
\end{align}
Moreover, $V(x^*(t_f))=\sum_{i=2}^N[\tilde{x}_i^*(t_f)]^TP\tilde{x}_i^*(t_f)
\geq \lambda_{\min}(P)\sum_{i=2}^N\Vert \tilde{x}_i^*(t_f) \Vert^2
\geq \lambda_{\min}(P)(N-1)\varepsilon^2.$
By \eqref{44}, the upper bound on the formation time is given by
\begin{align}
t_f&\leq \lambda_{\min}(P)\ln{\frac{V(x(0))}{V(x^*(t_f))}}\nonumber\\
&\leq\lambda_{\min}(P)\ln{\frac{V(x(0))}{\lambda_{\min}(P)(N-1)\varepsilon^2}}.\nonumber
\end{align}
Therefore, for the given steady-state error tolerance $\varepsilon$ and the termination time $T$, the formation task can be achieved if
\begin{align}
T\geq \lambda_{\min}(P)\ln{\frac{V(x(0))}{\lambda_{\min}(P)(N-1)\varepsilon^2}},\nonumber
\end{align}
where by \eqref{36} 
\begin{align}
&\lambda_{\min}(P)=\frac{1}{2\sqrt{\sigma\alpha}}\bigg(\bigg(1+\sqrt{\sigma\alpha}\bigg)\sqrt{1+\frac{\beta}{\alpha}+\frac{2}{\sqrt{\sigma\alpha}}}\nonumber\\
&-\sqrt{\bigg(1+\sigma\alpha-2\sqrt{\sigma\alpha}\bigg)\bigg(1+\frac{\beta}{\alpha}+\frac{2}{\sqrt{\sigma\alpha}}\bigg)+4}\bigg).\nonumber
\end{align}

Let $J_{E^*}$ denote the energy consumption during $[0,T]$ under the optimal control law (\ref{11}). Due to Assumption \ref{assump1}, $J_{E^*}$ can be written as $J_{E^*}=\sum_{i=2}^NJ_{\tilde{E}_i^*},$
where $$J_{\tilde{E}_i^*}=\int_0^{T}\{[\tilde{u}_i^*(t)]^T\tilde{u}_i^*(t)+\beta\lambda_i[\tilde{x}_i^*(t)]^TQ\tilde{x}_i^*(t)\}dt.$$
It follows from (\ref{34}) that
\begin{align}\label{51}
J_{\tilde{E}_i^*}&=\int_0^{T}[\tilde{x}_i^*(t)]^T(\alpha^2\lambda_i^2PBB^TP+\beta\lambda_iQ)\tilde{x}_i^*(t)dt\nonumber\\
&\leq[\alpha^2\lambda_i\lambda_{\max}(PBB^TP)+\beta]\lambda_i\int_0^{T}[\tilde{x}_i^*(t)]^T\tilde{x}_i^*(t)dt\nonumber\\
&\leq\bigg[\lambda_N\bigg(\alpha+\frac{1}{\sigma}\bigg)\bigg(\alpha+\beta+2\sqrt{\frac{\alpha}{\sigma}}\bigg)+\beta\bigg]\lambda_i\nonumber\\
&\quad \times\int_0^{T}[\tilde{x}_i^*(t)]^T\tilde{x}_i^*(t)dt.
\end{align}
On the other hand, the solution of (\ref{17}) is given by
\begin{align}
\tilde{x}_i^*(t)=e^{(A-\lambda_i\alpha BB^TP)t}\tilde{x}_i(0). \nonumber
\end{align}
It hence follows that $\int_0^{T}[\tilde{x}_i^*(t)]^T\tilde{x}_i^*(t)dt=\int_0^{T}\tilde{x}_i^T(0) \times e^{(A-\lambda_i\alpha BB^TP)^Tt}e^{(A-\lambda_i\alpha BB^TP)t}\tilde{x}_i(0)dt\leq \Vert\tilde{x}_i(0)\Vert^2\int_0^{T} $ $\Vert e^{(A-\lambda_i\alpha BB^TP)t}\Vert^2dt. $
Since
\begin{align}
\Vert e^{(A-\lambda_i\alpha BB^TP)t}\Vert\leq e^{[\max\limits_{i=2,\dots,N}\lambda_{\max}(A-\lambda_i\alpha BB^TP)]t}, \nonumber
\end{align}
it follows that
\begin{align}\label{55}
& \quad\int_0^{T}[\tilde{x}_i^*(t)]^T\tilde{x}_i^*(t)dt\nonumber\\
&\leq \Vert\tilde{x}_i(0)\Vert^2\int_0^{T}e^{2[\max\limits_{i=2,\dots,N}\lambda_{\max}(A-\lambda_i\alpha BB^TP)]t}dt\nonumber\\
&= \frac{\Vert\tilde{x}_i(0)\Vert^2}{2\vert\max\limits_{i=2,\dots,N}\lambda_{\max}(A-\lambda_i\alpha BB^TP)\vert} \nonumber\\
&\quad \times\bigg(1-e^{2[\max\limits_{i=2,\dots,N}\lambda_{\max}(A-\lambda_i\alpha BB^TP)]T}\bigg).
\end{align}
It can be verified that
\begin{align}
&\quad \lambda_{\max}(A-\lambda_i\alpha BB^TP) \nonumber \\
&=\frac{1}{2}\bigg(-\lambda_i\sqrt{\frac{\alpha}{\sigma}\bigg(1+\frac{\beta}{\alpha}+\frac{2}{\sqrt{\sigma\alpha}}\bigg)}\nonumber\\
&\quad +\sqrt{\lambda_i^2\frac{\alpha}{\sigma}\bigg(1+\frac{\beta}{\alpha}+\frac{2}{\sqrt{\sigma\alpha}}\bigg)-4\lambda_i\sqrt{\frac{\alpha}{\sigma}}}\bigg). \nonumber
\end{align}
Additionally,
\begin{align}\label{57}
& \quad \lambda_i\Vert \tilde{x}_i(0)\Vert^2 \nonumber \\
&\leq\sum_{i=1}^N\lambda_i \Vert \tilde{x}_i(0)\Vert^2\nonumber\\
&\leq [x(0)-x^d]^T(\mathcal{L}\otimes I_{2n})[x(0)-x^d].
\end{align}
Combining (\ref{51}), (\ref{55}), and (\ref{57}) leads to
\begin{align}\label{58}
J_{\tilde{E}_i^*} \leq & \frac{V_\mathcal{L}(0)[\lambda_N(\alpha+\frac{1}{\sigma})(\alpha+\beta+2\sqrt{\frac{\alpha}{\sigma}})+\beta]}{2\vert\max\limits_{i=2,\dots,N}\lambda_{\max}(A-\lambda_i\alpha BB^TP)\vert}\nonumber\\
&\times\bigg(1-e^{2[\max\limits_{i=2,\dots,N}\lambda_{\max}(A-\lambda_i\alpha BB^TP)]T}\bigg),
\end{align}
where
\begin{align*}
V_\mathcal{L}(0)=[x(0)-x^d]^T(\mathcal{L}\otimes I_{2n})[x(0)-x^d].
\end{align*}
Let $\mathrm{Max}(\lambda)$ denote the  parameter $\lambda_i$ that maximizes $\lambda_{\max}(A-\lambda_i\alpha BB^TP)$. Eq.~\eqref{58} can be written as
\begin{align}\label{59}
J_{\tilde{E}_i^*} \leq & \frac{V_\mathcal{L}(0)[\lambda_N(\alpha+\frac{1}{\sigma})(\alpha+\beta+2\sqrt{\frac{\alpha}{\sigma}})+\beta]}{2\vert\lambda_{\max}(A-\mathrm{Max}(\lambda)\alpha BB^TP)\vert}\nonumber\\
&\times\bigg(1-e^{2\lambda_{\max}(A-\mathrm{Max}(\lambda)\alpha BB^TP)T}\bigg).
\end{align}
Since
\begin{align}
 & \quad \lambda_{\max}(A-\mathrm{Max}(\lambda)\alpha BB^TP)\nonumber \\
&=\frac{1}{2}\bigg(-\mathrm{Max}(\lambda)\sqrt{\frac{\alpha}{\sigma}\bigg(1+\frac{\beta}{\alpha}
+\frac{2}{\sqrt{\sigma\alpha}}\bigg)}\nonumber\\
&\quad+\sqrt{\mathrm{Max}^2(\lambda)\frac{\alpha}{\sigma}\bigg(1+\frac{\beta}{\alpha}+\frac{2}{\sqrt{\sigma\alpha}}\bigg)- 4\mathrm{Max}(\lambda)\sqrt{\frac{\alpha}{\sigma}}}\bigg)\nonumber\\
&\geq -\frac{1}{2}\mathrm{Max}(\lambda)\sqrt{\frac{\alpha}{\sigma}\bigg(1+\frac{\beta}{\alpha}+\frac{2}{\sqrt{\sigma\alpha}}\bigg)}, \nonumber
\end{align}
it follows that
\begin{align}\label{60}
&\quad1-e^{2\lambda_{\max}(A-\mathrm{Max}(\lambda)\alpha BB^TP)T}\nonumber\\
&\leq1-e^{-\mathrm{Max}(\lambda)\sqrt{\frac{\alpha}{\sigma}(1+\frac{\beta}{\alpha}+\frac{2}{\sqrt{\alpha\sigma}})}T}\nonumber\\
&\leq1-e^{-\lambda_N\sqrt{\frac{\alpha}{\sigma}(1+\frac{\beta}{\alpha}+\frac{2}{\sqrt{\alpha\sigma}})}T}.
\end{align}
Combining \eqref{59} and \eqref{60} leads to
\begin{align}\label{61}
J_{\tilde{E}_i^*}&\leq \frac{V_\mathcal{L}(0)[\lambda_N(\alpha+\frac{1}{\sigma})(\alpha+\beta+2\sqrt{\frac{\alpha}{\sigma}})+\beta]}{2\vert\lambda_{\max}(A-\mathrm{Max}(\lambda)\alpha BB^TP)\vert}\nonumber\\
&\quad\times\bigg(1-e^{-\lambda_N\sqrt{\frac{\alpha}{\sigma}(1+\frac{\beta}{\alpha}+\frac{2}{\sqrt{\alpha\sigma}})}T}\bigg)\nonumber\\
&=  \frac{V_\mathcal{L}(0)[\lambda_N(\alpha+\frac{1}{\sigma})(\alpha+\beta+2\sqrt{\frac{\alpha}{\sigma}})+\beta]}{4\mathrm{Max}(\lambda)\sqrt{\frac{\alpha}{\sigma}}}\nonumber\\
&\quad\times\bigg(\sqrt{\mathrm{Max}^2(\lambda)\frac{\alpha}{\sigma}(1+\frac{\beta}{\alpha}+\frac{2}{\sqrt{\alpha\sigma}})-4\mathrm{Max}(\lambda)\sqrt{\frac{\alpha}{\sigma}}}\nonumber\\
&\quad +\mathrm{Max}(\lambda)\sqrt{\frac{\alpha}{\sigma}(1+\frac{\beta}{\alpha}+\frac{2}{\sqrt{\alpha\sigma}})} \bigg)\nonumber\\
&\quad\times\bigg(1-e^{-\lambda_N\sqrt{\frac{\alpha}{\sigma}(1+\frac{\beta}{\alpha}+\frac{2}{\sqrt{\alpha\sigma}})}T}\bigg),
\end{align}
which holds by multiplying the numerator and denominator with $\sqrt{\mathrm{Max}^2(\lambda)\frac{\alpha}{\sigma}\bigg(1+\frac{\beta}{\alpha}+\frac{2}{\sqrt{\sigma\alpha}}\bigg)-4\mathrm{Max}(\lambda)\sqrt{\frac{\alpha}{\sigma}}}
+\mathrm{Max}(\lambda)\sqrt{\frac{\alpha}{\sigma}\bigg(1+\frac{\beta}{\alpha}+\frac{2}{\sqrt{\sigma\alpha}}\bigg)}.$
Additionally,
\begin{align}\label{62}
&\quad\sqrt{\mathrm{Max}^2(\lambda)\frac{\alpha}{\sigma}\bigg(1+\frac{\beta}{\alpha}+\frac{2}{\sqrt{\sigma\alpha}}\bigg)-4\mathrm{Max}(\lambda)\sqrt{\frac{\alpha}{\sigma}}}\nonumber\\
&\quad+\mathrm{Max}(\lambda)\sqrt{\frac{\alpha}{\sigma}\bigg(1+\frac{\beta}{\alpha}+\frac{2}{\sqrt{\sigma\alpha}}\bigg)}\nonumber\\
&\leq2\mathrm{Max}(\lambda)\sqrt{\frac{\alpha}{\sigma}\bigg(1+\frac{\beta}{\alpha}+\frac{2}{\sqrt{\sigma\alpha}}\bigg)}.
\end{align}
Substituting \eqref{62} into \eqref{61} yields
\begin{align*}
&\quad J_{\tilde{E}_i^*}\\
&\leq\frac{V_\mathcal{L}(0)[\lambda_N(\alpha+\frac{1}{\sigma})(\alpha+\beta+2\sqrt{\frac{\alpha}{\sigma}})+\beta]}{4\mathrm{Max}(\lambda)\sqrt{\frac{\alpha}{\sigma}}}2\mathrm{Max}(\lambda)\\
&\times\sqrt{\frac{\alpha}{\sigma}\bigg(1+\frac{\beta}{\alpha}+\frac{2}{\sqrt{\sigma\alpha}}\bigg)}\bigg(1-e^{-\lambda_N\sqrt{\frac{\alpha}{\sigma}(1+\frac{\beta}{\alpha}+\frac{2}{\sqrt{\alpha\sigma}})}T}\bigg)\\
&=\frac{1}{2}V_\mathcal{L}(0) \bigg[\lambda_N\bigg(\alpha+\frac{1}{\sigma}\bigg)\bigg(\alpha+\beta+2\sqrt{\frac{\alpha}{\sigma}}\bigg)+\beta\bigg]\\
&\times\sqrt{1+\frac{\beta}{\alpha}+\frac{2}{\sqrt{\alpha\sigma}}}\bigg(1-e^{-\lambda_N\sqrt{\frac{\alpha}{\sigma}(1+\frac{\beta}{\alpha}+\frac{2}{\sqrt{\alpha\sigma}})}T}\bigg).
\end{align*}
The energy constraint is given by $J_{\tilde{E}_i^*}\leq E_i(0).$
Thus, the energy requirement can be met if
\begin{align}
E_i(0)&\geq \frac{1}{2}V_\mathcal{L}(0) \bigg[\lambda_N\bigg(\alpha+\frac{1}{\sigma}\bigg)\bigg(\alpha+\beta+2\sqrt{\frac{\alpha}{\sigma}}\bigg)+\beta\bigg] \nonumber\\ &\times\sqrt{1+\frac{\beta}{\alpha}+\frac{2}{\sqrt{\alpha\sigma}}}\bigg(1-e^{-\lambda_N\sqrt{\frac{\alpha}{\sigma}(1+\frac{\beta}{\alpha}+\frac{2}{\sqrt{\alpha\sigma}})}T}\bigg),\nonumber\\
& \quad i\in\{1,\dots,N\}. \nonumber
\end{align}
The proof is thus completed.
\end{proof}

According to Theorem~\ref{the1}, if the time constraint is removed, i.e., $T\rightarrow\infty$, the energy bound can be simplified as $E_i(0)\geq \frac{1}{2}V_\mathcal{L}(0) \bigg[\lambda_N\bigg(\alpha+\frac{1}{\sigma}\bigg)\bigg(\alpha+\beta+2\sqrt{\frac{\alpha}{\sigma}}\bigg)+\beta\bigg]
\sqrt{1+\frac{\beta}{\alpha}+\frac{2}{\sqrt{\alpha\sigma}}},\quad i\in\{1,\dots,N\}.$
Additionally, it is noted that if the initial formation error is large, a longer termination time $T$ and a higher energy level $E_i(0)$ are expected for achieving the formation of the multi-agent system. Besides, the smaller the formation threshold $\varepsilon$, the longer the termination time and the more the energy consumption.

\section{Monotonicity properties of the optimal formation algorithm}

\label{sec:IV}

This section is devoted to the discussion of the relationships between the lower bound of the required initial energy $E_i(0)$, the lower bound of the achievable termination time $T$, and the algorithm parameters.

\subsection{Monotonicity of the PARE solution}
The following result presents the monotonicity of the solution $P$ of the PARE \eqref{33} with respect to the parameters $\alpha$, $\sigma$, and $\beta$.
\begin{theorem}\label{the2}
The solution $P$ of the PARE \eqref{33} is a decreasing function of $\alpha$ and $\sigma$, and an increasing function of $\beta$, i.e.,
\begin{align}
\frac{\partial P}{\partial \alpha}\leq 0,\quad \frac{\partial P}{\partial \sigma}\leq 0,\quad \frac{\partial P}{\partial \beta}\geq 0, \quad \forall {\alpha, \sigma, \beta}>0. \nonumber
\end{align}
\end{theorem}

\begin{proof}
It follows from \eqref{33} that
\begin{align}\label{64}
&\quad(A-\alpha\sigma BB^TP)^TP+P(A-\alpha\sigma BB^TP)\nonumber\\
&=-\bigg(I_{2n}+\frac{\beta}{\alpha}Q+\alpha\sigma PBB^TP\bigg).
\end{align}
Since $I_{2n}+\frac{\beta}{\alpha}Q+\alpha\sigma PBB^TP$ is positive definite, it follows from \eqref{64} that $(A-\alpha\sigma BB^TP)$ is Hurwitz. To show the relationship between $P$ and $\alpha$, differentiating both sides of \eqref{64} with respect to $\alpha$ yields
\begin{align}\label{65}
&\quad \frac{\partial P}{\partial\alpha}(A-\alpha\sigma BB^TP)+(A-\alpha\sigma BB^TP)^T\frac{\partial P}{\partial \alpha}\nonumber\\
&=\sigma PBB^TP+\frac{\beta}{\alpha^2}Q.
\end{align}
Since $(A-\alpha\sigma BB^TP)$ is Hurwitz, and the right-hand side of \eqref{65} is positive semidefinite, \eqref{65} has the following unique solution
\begin{align}
\frac{\partial P}{\partial \alpha}&=-\int_0^\infty e^{(A-\alpha\sigma BB^TP)^Tt}\bigg(\sigma PBB^TP+\frac{\beta}{\alpha^2}Q\bigg)\nonumber\\
&\quad \times e^{(A-\alpha\sigma BB^TP)t}dt\nonumber\\
&\leq 0. \nonumber
\end{align}
Thus, $P$ is monotonically decreasing with $\alpha$.
Similarly, it can be shown that
\begin{align}
\frac{\partial P}{\partial\sigma}(A-\alpha\sigma BB^TP)+(A-\alpha\sigma BB^TP)^T\frac{\partial P}{\partial \sigma}=\alpha PBB^TP, \nonumber
\end{align}
which has the following unique solution
\begin{align}
\frac{\partial P}{\partial \sigma}&=-\int_0^\infty e^{(A-\alpha\sigma BB^TP)^Tt}\alpha PBB^TPe^{(A-\alpha\sigma BB^TP)t}dt\nonumber\\
&\leq 0. \nonumber
\end{align}
Similarly, it can be shown that
\begin{align}
\frac{\partial P}{\partial\beta}(A-\alpha\sigma BB^TP)+(A-\alpha\sigma BB^TP)^T\frac{\partial P}{\partial \beta}=-\frac{1}{\alpha}Q, \nonumber
\end{align}
which has the unique solution
\begin{align}
\frac{\partial P}{\partial \beta}&=\int_0^\infty e^{(A-\alpha\sigma BB^TP)^Tt}\frac{1}{\alpha}Qe^{(A-\alpha\sigma BB^TP)t}dt\geq 0. \nonumber
\end{align}
Thus, $P$ is monotonically decreasing with $\sigma$ and monotonically increasing with $\beta$. The proof is thus completed.
\end{proof}

\subsection{Termination time}
\label{sect:time}
The following result discusses the monotonicity of the lower bound of the termination time $T$ in \eqref{12}. 

\begin{theorem}\label{the3}
The lower bound of the achievable termination time $T$ in (\ref{12}) is a decreasing function of both $\sigma$ and $\alpha$ and an increasing function of $\beta$.
\end{theorem}
\begin{proof}
For notational convenience, define the lower bound of the termination time $T$ as $T_l$, i.e.,
\begin{align}\label{71}
T_l=\lambda_{\min}(P)\ln{\frac{V(x(0))}{\lambda_{\min}(P)(N-1)\varepsilon^2}}.
\end{align}
Differentiating both sides of \eqref{71} with respect to $\alpha$ yields
\begin{align}\label{72}
\frac{\partial {T}_l}{\partial \alpha}=&\frac{\partial\lambda_{\min}(P)}{\partial\alpha}\bigg(\ln{\frac{V(x(0))}{\lambda_{\min}(P)(N-1)\varepsilon^2}}-1\bigg)\nonumber\\
&+\frac{\lambda_{\min}(P)}{V(x(0))}\frac{\partial V(x(0))}{\partial \alpha}.
\end{align}
It is straightforward to know that $\frac{\partial T_l}{\partial \alpha}\leq 0$ if the two terms on the right-hand side of \eqref{72} are non-positive. According to the relationship of $P$ and $\alpha$, it follows that
\begin{align}
\frac{\partial V(x(0))}{\partial \alpha}\leq 0. \nonumber
\end{align}
Since $P$ is symmetric and positive semidefinite, it can be diagonalized as
\begin{align}\label{74}
P=\mathcal{M}^T\Lambda(P)\mathcal{M},
\end{align}
where $\mathcal{M}=[m_1,\dots,m_{2n}]$ is the matrix comprising the orthonormal eigenvectors of $P$ and $\Lambda(P)=\mathrm{diag}([\lambda_1(P),\dots,\lambda_{2n}(P)])$ with $\lambda_i(P)$ being the $i$th eigenvalue of $P$.\\
Differentiating both sides of \eqref{74} with respect to $\alpha$ yields
\begin{align}
\frac{\partial P}{\partial \alpha}&=\frac{\partial \mathcal{M}^T}{\partial \alpha}(\Lambda(P)\mathcal{M})+\mathcal{M}^T\frac{\partial (\Lambda(P)\mathcal{M})}{\partial \alpha}\nonumber\\
&=0+\mathcal{M}^T(\frac{\partial \Lambda(P)}{\partial \alpha}\mathcal{M}+\Lambda(P)\frac{\partial \mathcal{M}}{\partial \alpha})\nonumber\\
&=\mathcal{M}^T\frac{\partial \Lambda(P)}{\partial \alpha}\mathcal{M}. \nonumber
\end{align}
Since $\frac{\partial P}{\partial \alpha}\leq 0$, each eigenvalue of $\frac{\partial P}{\partial \alpha}$ must be non-positive, i.e.,
\begin{align}
\frac{\partial \lambda_i(P)}{\partial \alpha}\leq 0,\quad i\in\{1,\dots,2n\}, \nonumber
\end{align}
which gives $\frac{\partial\lambda_{\min}(P)}{\partial \alpha}\leq 0.$
Additionally,
\begin{align*}
&\quad\ln{\frac{V(x(0))}{\lambda_{\min}(P)(N-1)\varepsilon^2}}-1 \\ &\geq \ln{\frac{\lambda_{\min}(P)\sum_{i=2}^N\Vert\tilde{x}_i(0)\Vert^2}{\lambda_{\min}(P)(N-1)\varepsilon^2}}-1\\
&=\ln{\frac{\sum_{i=2}^N\Vert\tilde{x}_i(0)\Vert^2}{(N-1)\varepsilon^2}}-1\geq 0,
\end{align*}
which leads to $\frac{\partial {T}_l}{\partial \alpha}\leq 0.$
Similarly, it can be shown that
\begin{align*}
\frac{\partial {T}_l}{\partial \sigma}\leq 0, \quad \frac{\partial {T}_l}{\partial \beta}\geq 0.
\end{align*}
The proof is thus completed.
\end{proof}

\subsection{Energy expenditure}
\label{sect:ener}
Next, the effect of the parameters $\alpha$, $\sigma$, and $\beta$ on the lower bound of the energy level $E_i(0)$ in \eqref{13} is investigated. The following assumption is made in this subsection.
\begin{assumption}\label{assump2}
Suppose that
\begin{align*}
\lambda_N \left(\frac{3}{2}\alpha+\frac{1}{2}\beta+2\sqrt{\frac{\alpha}{\sigma}}+\frac{1}{2\sigma}-\frac{\beta}{2\alpha\sigma}\right)-\frac{\beta}{2\alpha}\geq 0,
\end{align*}
where $\lambda_N$ denotes the largest eigenvalue of the Laplacian matrix.
\end{assumption}

\begin{theorem}\label{the4}
If Assumption \ref{assump2} holds, then the lower bound of the required initial energy $E_i(0)$ in \eqref{13} is an increasing function of $\alpha$ and $\beta$ and a decreasing function of $\sigma$.
\end{theorem}

\begin{proof}
For notational convenience, define the lower bound of $E_i(0)$ as $E_{i_l}$, i.e.,
\begin{align}\label{80}
&E_{i_l}=\frac{1}{2}V_\mathcal{L}(0)\bigg[\lambda_N\bigg(\alpha+\frac{1}{\sigma}\bigg)\bigg(\alpha+\beta+2\sqrt{\frac{\alpha}{\sigma}}\bigg)+\beta\bigg]\nonumber\\
&\times\sqrt{1+\frac{\beta}{\alpha}+\frac{2}{\sqrt{\alpha\sigma}}}\bigg(1-e^{-\lambda_N\sqrt{\frac{\alpha}{\sigma}(1+\frac{\beta}{\alpha}+\frac{2}{\sqrt{\alpha\sigma}})}T}\bigg).
\end{align}
Differentiating both sides of \eqref{80} with respect to $\alpha$ yields
\begin{align}\label{81}
\frac{\partial {E_{i_l}}}{\partial \alpha}= \frac{\partial H_1}{\partial \alpha}V_\mathcal{L}(0)H_2+\frac{\partial H_2}{\partial \alpha}V_\mathcal{L}(0)H_1,
\end{align}
where $H_1=\bigg[\lambda_N\bigg(\alpha+\frac{1}{\sigma}\bigg)\bigg(\alpha+\beta+2\sqrt{\frac{\alpha}{\sigma}}\bigg)+\beta\bigg] \sqrt{1+\frac{\beta}{\alpha}+\frac{2}{\sqrt{\alpha\sigma}}},$ $H_2=\frac{1}{2}\bigg(1-e^{-\lambda_N\sqrt{\frac{\alpha}{\sigma}(1+\frac{\beta}{\alpha}+\frac{2}{\sqrt{\alpha\sigma}})}T}\bigg)$,
$\frac{\partial H_1}{\partial \alpha}=\lambda_N\bigg(2\alpha+\beta+3\sqrt{\frac{\alpha}{\sigma}}+\frac{1}{\sigma}+\frac{1}{\sigma\sqrt{\sigma\alpha}}\bigg) 
 \sqrt{1+\frac{\beta}{\alpha}+\frac{2}{\sqrt{\alpha\sigma}}}-\bigg[\lambda_N\bigg(\alpha+\frac{1}{\sigma}\bigg)\bigg(\alpha+\beta
+2\sqrt{\frac{\alpha}{\sigma}}\bigg)+\beta\bigg] \frac{\frac{\beta}{\alpha^2}+\frac{1}{\alpha\sqrt{\alpha\sigma}}}{2\sqrt{1+\frac{\beta}{\alpha}+\frac{2}{\sqrt{\alpha\sigma}}}},$ and $\frac{\partial H_2}{\partial \alpha}=\lambda_NT\frac{\frac{1}{\sigma}(1+\frac{1}{\sqrt{\sigma\alpha}})}{4\sqrt{\frac{\alpha}{\sigma}(1+\frac{\beta}{\alpha}+\frac{2}{\sqrt{\alpha\sigma}})}}e^{-\lambda_N\sqrt{\frac{\alpha}{\sigma}(1+\frac{\beta}{\alpha}+\frac{2}{\sqrt{\alpha\sigma}})}T}$.
It can be seen that $\frac{\partial {E_{i_l}}}{\partial \alpha}\geq 0$ when the two terms on the right-hand side of (\ref{81}) are non-negative. Since $V_\mathcal{L}(0)\geq 0,\quad H_1\geq 0,\quad H_2\geq 0$, and  $\frac{\partial H_2}{\partial \alpha}\geq 0,$
the second term of (\ref{81}) is non-negative. In the following, the sign of $\frac{\partial H_1}{\partial \alpha}$ is discussed. It follows that
\begin{align}
\frac{\partial H_1}{\partial \alpha}&\geq\lambda_N\bigg(2\alpha+\beta+3\sqrt{\frac{\alpha}{\sigma}}+\frac{1}{\sigma}+\frac{1}{\sigma\sqrt{\sigma\alpha}}\bigg)\nonumber\\
&\times\sqrt{1+\frac{\beta}{\alpha}+\frac{2}{\sqrt{\alpha\sigma}}}-\bigg[\lambda_N\bigg(\alpha+\frac{1}{\sigma}\bigg)\bigg(\alpha+\beta\nonumber\\
&+2\sqrt{\frac{\alpha}{\sigma}}\bigg)+\beta\bigg]
\frac{\frac{1}{\alpha}[(\frac{\beta}{\alpha}+\frac{1}{\sqrt{\alpha\sigma}})+1+\frac{1}{\sqrt{\alpha\sigma}}]}{2\sqrt{1+\frac{\beta}{\alpha}+\frac{2}{\sqrt{\alpha\sigma}}}},
\end{align}
which leads to 
$\frac{\partial H_1}{\partial \alpha} \geq  \lambda_N\bigg(2\alpha+\beta+3\sqrt{\frac{\alpha}{\sigma}}+\frac{1}{\sigma}+\frac{1}{\sigma\sqrt{\sigma\alpha}}\bigg)  \sqrt{1+\frac{\beta}{\alpha}+\frac{2}{\sqrt{\alpha\sigma}}}-\bigg[\lambda_N\bigg(\alpha+\frac{1}{\sigma}\bigg)\bigg(\alpha+\beta+2\sqrt{\frac{\alpha}{\sigma}}\bigg)+\beta\bigg] \frac{\frac{1}{\alpha}(1+\frac{\beta}{\alpha}+\frac{2}{\sqrt{\alpha\sigma}})}{2\sqrt{1+\frac{\beta}{\alpha}+\frac{2}{\sqrt{\alpha\sigma}}}}= \bigg[\lambda_N\bigg(\frac{3}{2}\alpha+\frac{1}{2}\beta+2\sqrt{\frac{\alpha}{\sigma}}+\frac{1}{2\sigma}-\frac{\beta}{2\alpha\sigma}\bigg)-\frac{\beta}{2\alpha}\bigg]\sqrt{1+\frac{\beta}{\alpha}+\frac{2}{\sqrt{\alpha\sigma}}}.$
When Assumption \ref{assump2} holds, one has $ \frac{\partial H_1}{\partial \alpha}\geq 0,$
which gives $\frac{\partial {E_{i_l}}}{\partial \alpha}\geq 0.$
Thus, the lower bound of the required initial energy $E_i(0)$ is an increasing function of $\alpha$.

Similarly, it can be shown that
\begin{align}
\frac{\partial {E_{i_l}}}{\partial \sigma}=V_\mathcal{L}(0)\frac{\partial H_1}{\partial \sigma}H_2+V_\mathcal{L}(0)\frac{\partial H_2}{\partial \sigma}H_1, \nonumber
\end{align}
where
\begin{align*}
\frac{\partial H_1}{\partial \sigma}&=-\frac{\lambda_N}{\sigma^2}\bigg(\alpha+\beta+3\sqrt{\frac{\alpha}{\sigma}}+\alpha\sqrt{\alpha\sigma}\bigg)\sqrt{1+\frac{\beta}{\alpha}+\frac{2}{\sqrt{\sigma\alpha}}}\\
&\quad -\frac{\lambda_N(\alpha+\frac{1}{\sigma})(\alpha+\beta+2\sqrt{\frac{\alpha}{\sigma}})+\beta}{2\sigma\sqrt{\alpha\sigma}\sqrt{1+\frac{\beta}{\alpha}+\frac{2}{\alpha\sigma}}} \nonumber \\
& \leq 0,\\
\frac{\partial H_2}{\partial \sigma}&=-\frac{\lambda_NT\frac{\alpha}{\sigma^2}(1+\frac{\beta}{\alpha}+\frac{3}{\sqrt{\alpha\sigma}})}{4\sqrt{1+\frac{\beta}{\alpha}+\frac{2}{\sqrt{\alpha\sigma}}}}e^{-\lambda_N\sqrt{\frac{\alpha}{\sigma}(1+\frac{\beta}{\alpha}+\frac{2}{\sqrt{\alpha\sigma}})}T}\nonumber \\
& \leq 0,
\end{align*}
which further leads to $\frac{\partial {E_{i_l}}}{\partial \sigma}\leq 0.$
Also, one has $\frac{\partial {E_{i_l}}}{\partial \beta}=V_\mathcal{L}(0)\frac{\partial H_1}{\partial \beta}H_2+V_\mathcal{L}(0)\frac{\partial H_2}{\partial \beta}H_1,$ and
\begin{align*}
\frac{\partial H_1}{\partial \beta}&=[\lambda_N(\alpha+\frac{1}{\sigma})+1]\sqrt{1+\frac{\beta}{\alpha}+\frac{2}{\sqrt{\alpha\sigma}}}\\
&\quad+\frac{[\lambda_N(\alpha+\frac{1}{\sigma})(\alpha+\beta+2\sqrt{\frac{\alpha}{\sigma}})+\beta]}{2\alpha\sqrt{1+\frac{\beta}{\alpha}+\frac{2}{\sqrt{\alpha\sigma}}}}\geq0,\\
\frac{\partial H_2}{\partial \beta}&=\frac{\lambda_NT}{4\sigma\sqrt{\frac{\alpha}{\sigma}(1+\frac{\beta}{\alpha}+\frac{2}{\sqrt{\alpha\sigma}})}}e^{-\lambda_N\sqrt{\frac{\alpha}{\sigma}(1+\frac{\beta}{\alpha}+\frac{2}{\sqrt{\alpha\sigma}})}T}\geq0,
\end{align*}
which leads to $\frac{\partial {E_{i_l}}}{\partial \beta}\geq 0.$
Thus, the lower bound of the required initial energy $E_i(0)$ is a decreasing function of $\sigma$ and an increasing function of $\beta$. The proof is hence completed.
\end{proof}

It follows from Theorems~\ref{the3} and \ref{the4} that the lower bounds on the achievable termination time and the required initial energy are both decreasing functions of $\sigma$. Hence, one can increase the value of $\sigma$ to reduce the formation time and the energy consumption. However, $\sigma$ is not allowed to be arbitrarily large, because the condition $\sigma< \lambda_2$ must be met as indicated by Theorem~\ref{the1}. Meanwhile, a large value of $\alpha$ is capable of speeding the convergence of the formation algorithm, yet at the cost of more energy consumption. Finally, the resistance coefficient $\beta$ is both harmful to convergence time as well as energy consumption. That is, a larger value of $\beta$ will lead to a longer convergence time and more energy consumption.

\section{Simulation}
\label{sect:sim}
In this section, numerical examples are presented to verify the theoretical results. Let $N=5$ and $n=2$. The initial states of the agents are given by $x_1(0)=(0,4,0,0)$, $x_2(0)=(12,9,0,0)$, $x_3(0)=(5,3,0,0)$, $x_4(0)=(9,3,0,0)$, and $x_5(0)=(4,0,0,0)$. The desired relative states are set to $d_{12}=(5,-2.5,0,0)$, $d_{23}=(5,2.5,0,0)$, $d_{34}=(-5,2.5,0,0)$, $d_{14}=(-5,-2.5,0,0), d_{15}=(5,0,0,0)$, and $d_{53}=(5,0,0,0)$. The initial energy levels are given by $E(0)=\{1000,1200,700,900,500\}$, the termination time is $T=3s$, and the steady-state error tolerance  is $\varepsilon=0.1$. The network topology is given in Fig.~\ref{Fig.1}, for which the eigenvalues of the Laplacian matrix $\mathcal{L}$ are $\mathrm{spec}(\mathcal{L})=\{0,1.382,1.382,3.618,3.618\}$, and the second smallest eigenvalues is $\lambda_2=1.382$.
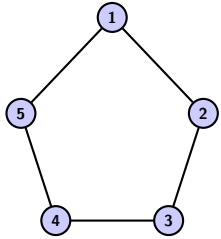
\begin{figure}[htp!]
\centering
\begin{tikzpicture}[-,auto,node distance=1.5cm, thick,main node/.style={circle,scale=.6,fill=blue!20,draw,font=\sffamily\bfseries}]
  \node[main node] (1) at(2.25,2.7027) {1};
  \node[main node] (2) at(3.4615,1.42665) {2};
  \node[main node] (3) at(3,0) {3};
  \node[main node] (4) at(1.5,0) {4};
  \node[main node] (5) at(1.0365,1.42665) {5};	
  \path[every node/.style={font=\sffamily\small}]
    (1) edge node [right] {} (2)
    (2) edge node [right] {} (3)
    (3) edge node [right] {} (4)
    (4) edge node [right] {} (5)
    (5) edge node [right] {} (1);
\end{tikzpicture}
\caption{The network topology}
\label{Fig.1}
\end{figure}

\begin{figure*}[htp!]
\centering
\setcounter{figure}{1}
\subfloat[$T_l$ vs. $\alpha$ ($\beta=0.2$)]{\label{Fig.2(a)}\includegraphics[width=0.3\textwidth]{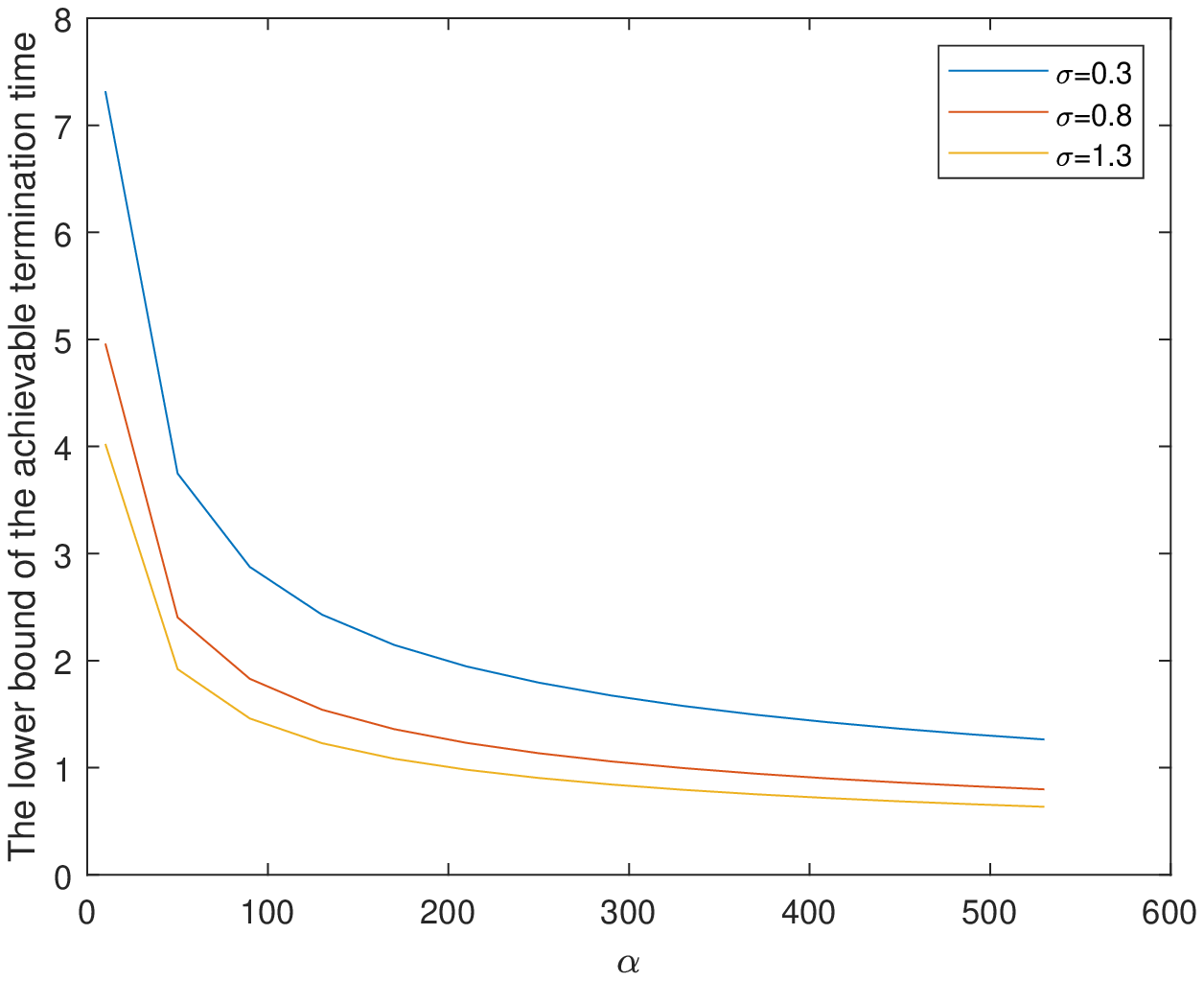}}
\quad
\subfloat[$T_l$ vs. $\sigma$ ($\beta=0.2$)]{\label{Fig.2(b)}\includegraphics[width=0.3\textwidth]{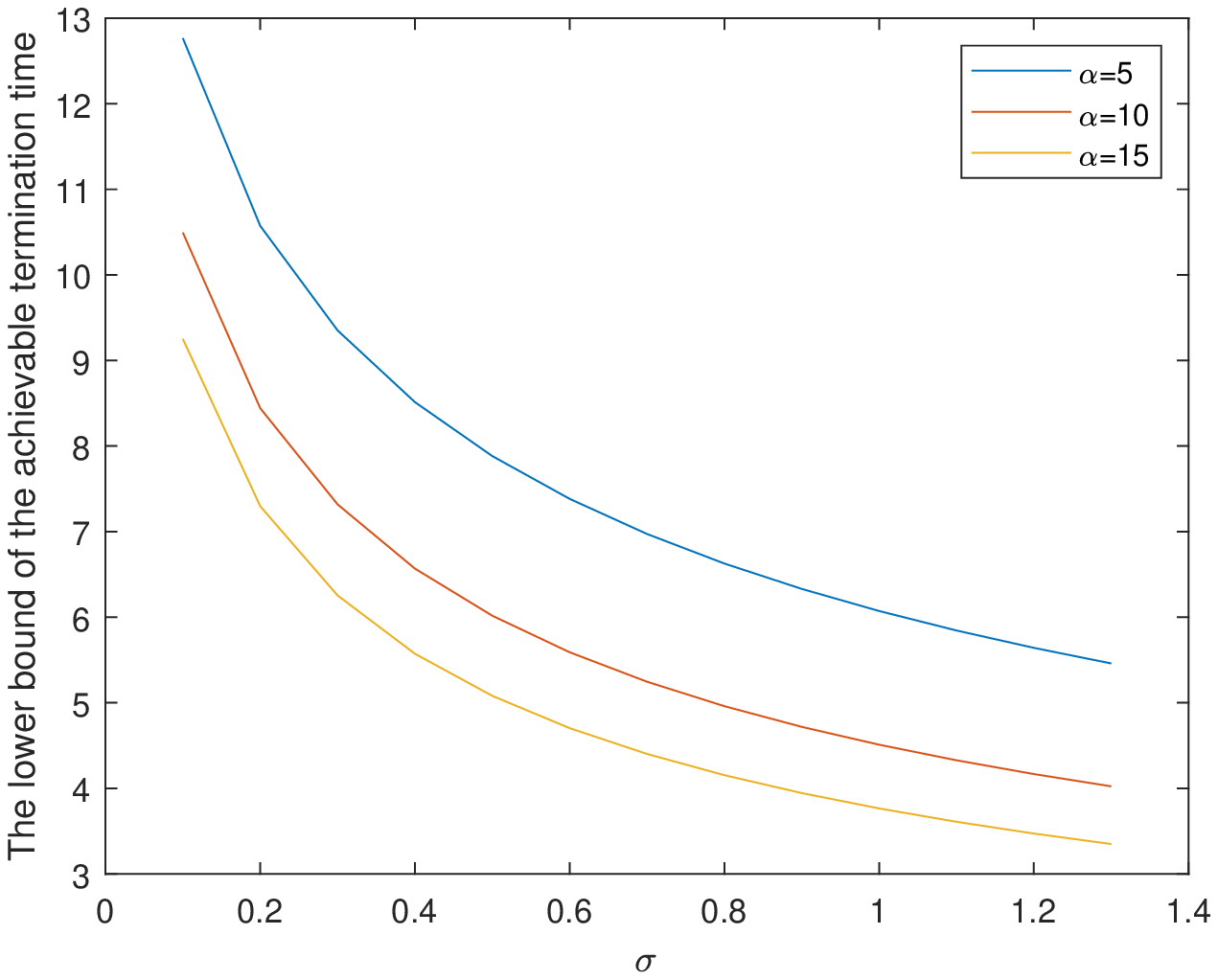}}
\quad
\subfloat[$T_l$ vs. $\beta$ ($\sigma=1.3$)]{\label{Fig.2(c)}\includegraphics[width=0.3\textwidth]{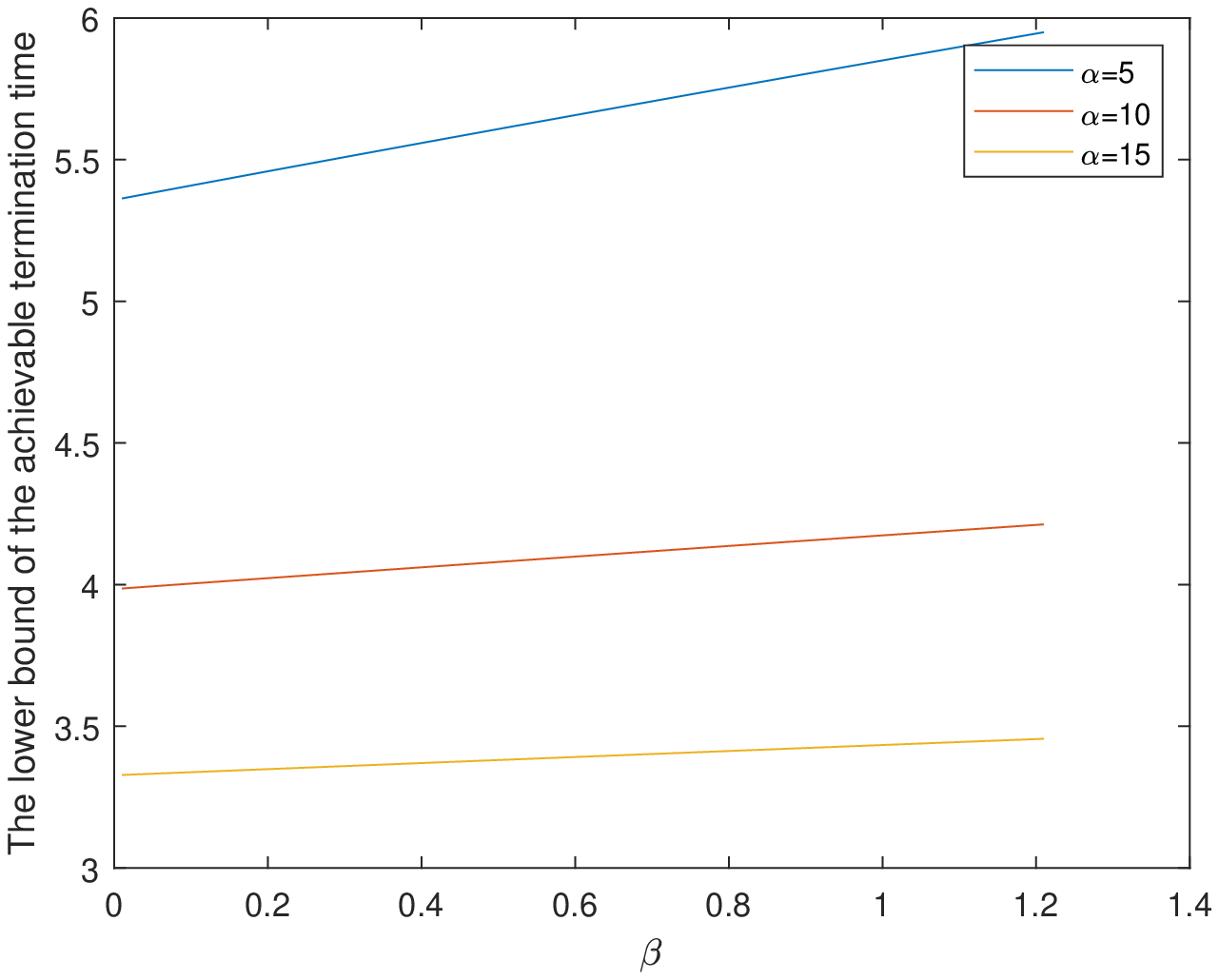}}\\
\caption{The lower bound of the achievable termination time for the multi-agent system.}
\label{Fig.2}
\end{figure*}

\begin{figure*}[htp!]
\centering
\setcounter{figure}{2}
\subfloat[$E_l$ vs. $\alpha$ ($\beta=0.2$)]{\label{Fig.3(a)}\includegraphics[width=0.3\textwidth]{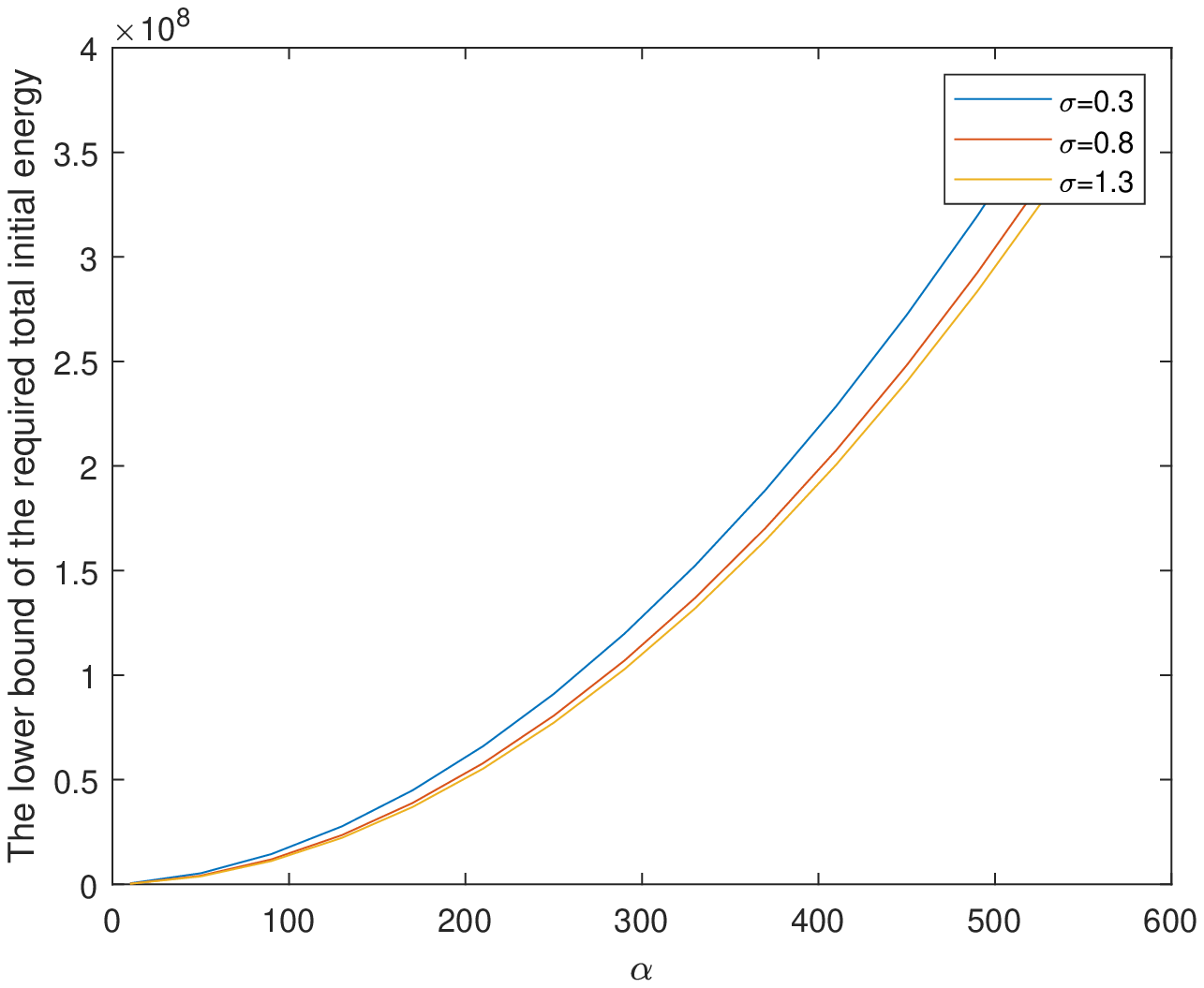}}
\quad
\subfloat[$E_l$ vs. $\sigma$ ($\beta=0.2$)]{\label{Fig.3(b)}\includegraphics[width=0.3\textwidth]{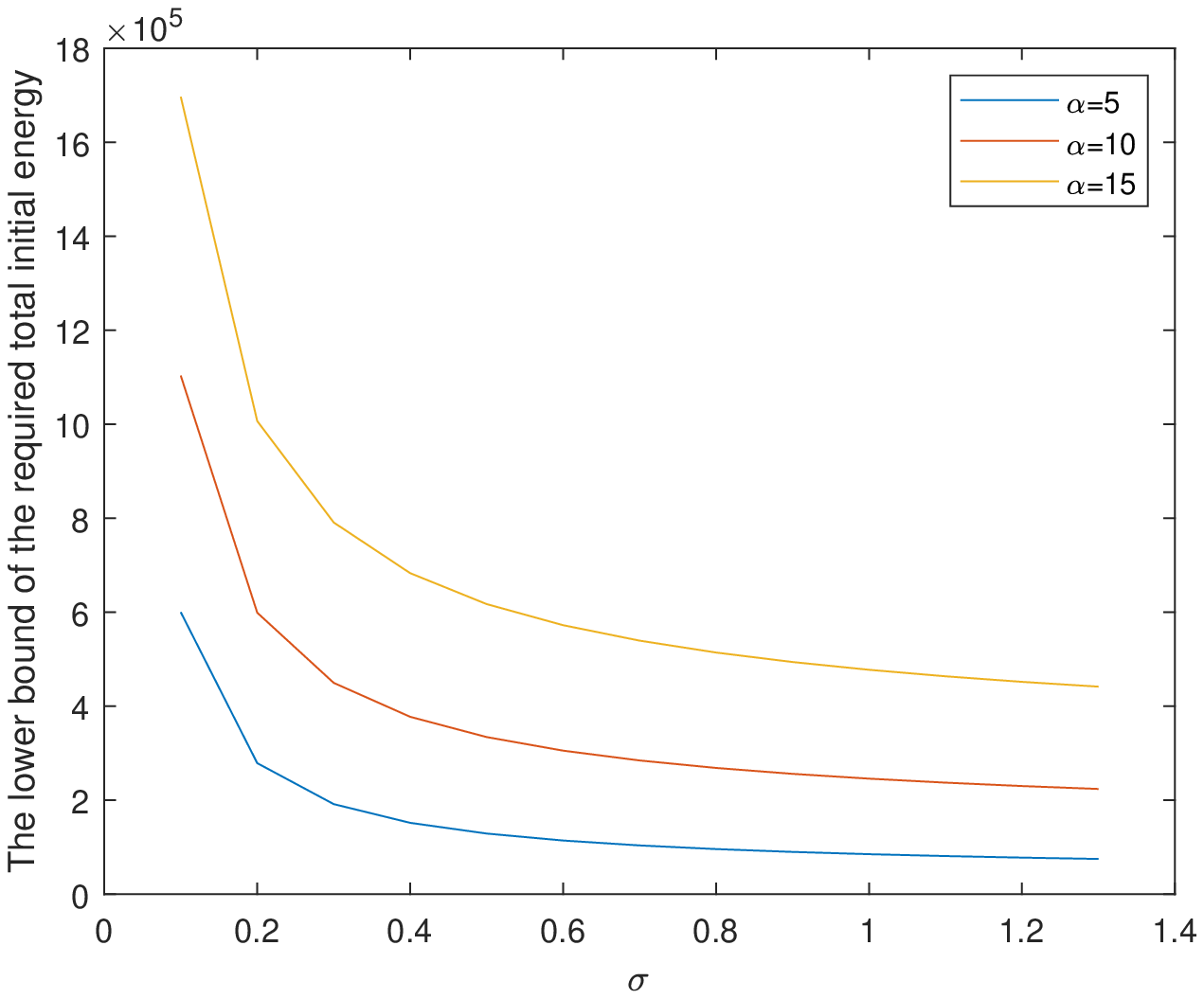}}
\quad
\subfloat[$E_l$ vs. $\beta$ ($\sigma=1.3$)]{\label{Fig.3(c)}\includegraphics[width=0.3\textwidth]{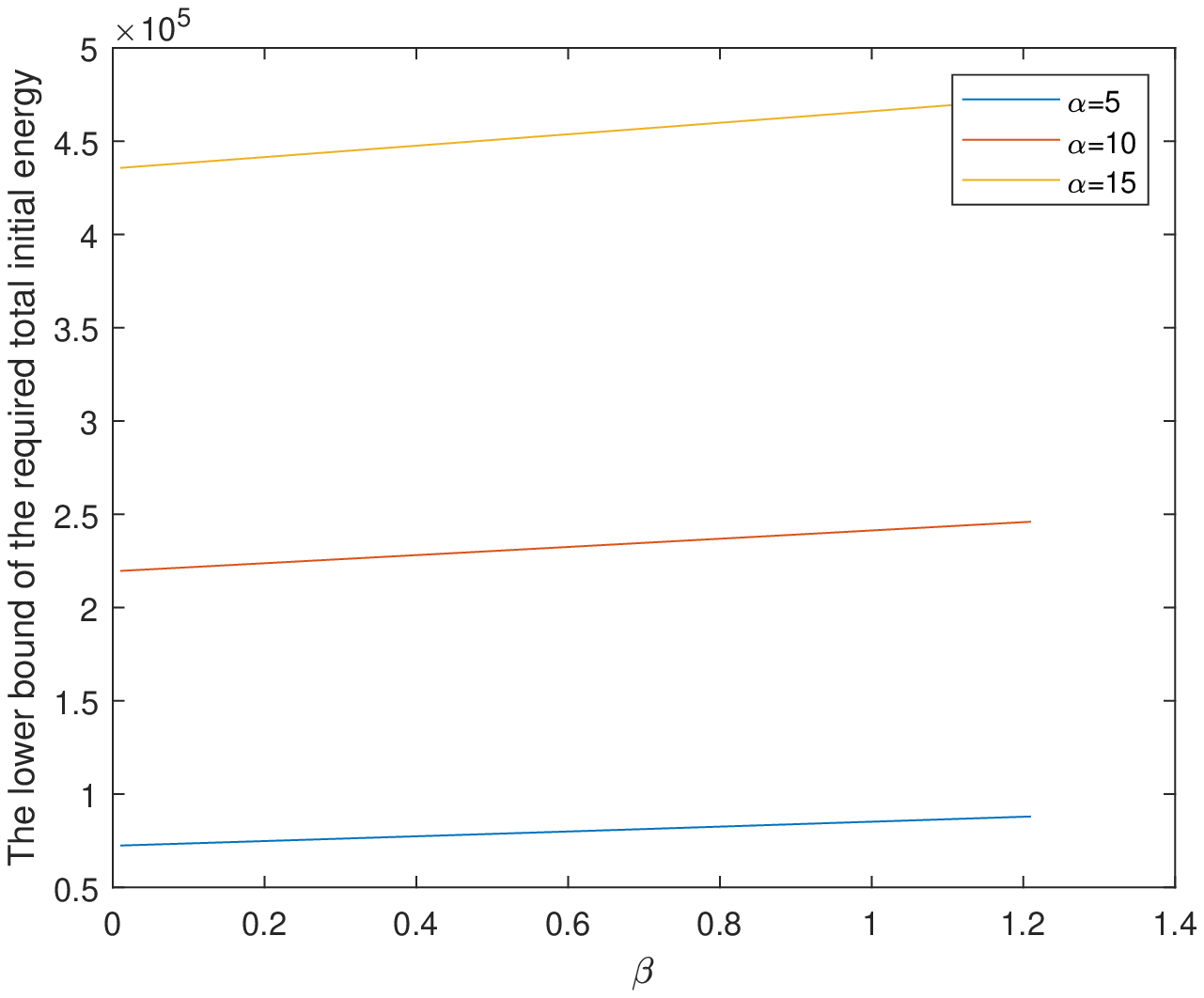}}\\
\caption{The lower bound of the required initial energy for all agents.}
\label{Fig.3}
\end{figure*}

Fig.~\ref{Fig.2} shows the curves of the lower bound of the achievable termination time versus the parameters $\alpha$, $\sigma$ and $\beta$. It can be observed that the lower bound of the achievable termination time is a decreasing function of both $\alpha$ and $\sigma$ and is an increasing function of $\beta$. This is consistent with the theoretical results in Section~\ref{sec:IV}. 
Fig.~\ref{Fig.3} shows the curves of the lower bound of the required total initial energy, i.e., $E_l=\sum_{i=1}^NE_{i_l}$, versus the parameters $\alpha$, $\sigma$ and $\beta$. It can be observed that the lower bound of the required total initial energy is an increasing function of $\alpha$ and $\beta$ and is a decreasing function of $\sigma$. In the following simulation, $\sigma=1.3$ is employed which is smaller than $\lambda_2$.

\begin{figure*}[htp!]
\centering
\setcounter{figure}{3}
\subfloat[Simulation I: $\alpha=450$, $\beta=0.2$]{\label{Fig.4(a)}\includegraphics[width=0.3\textwidth]{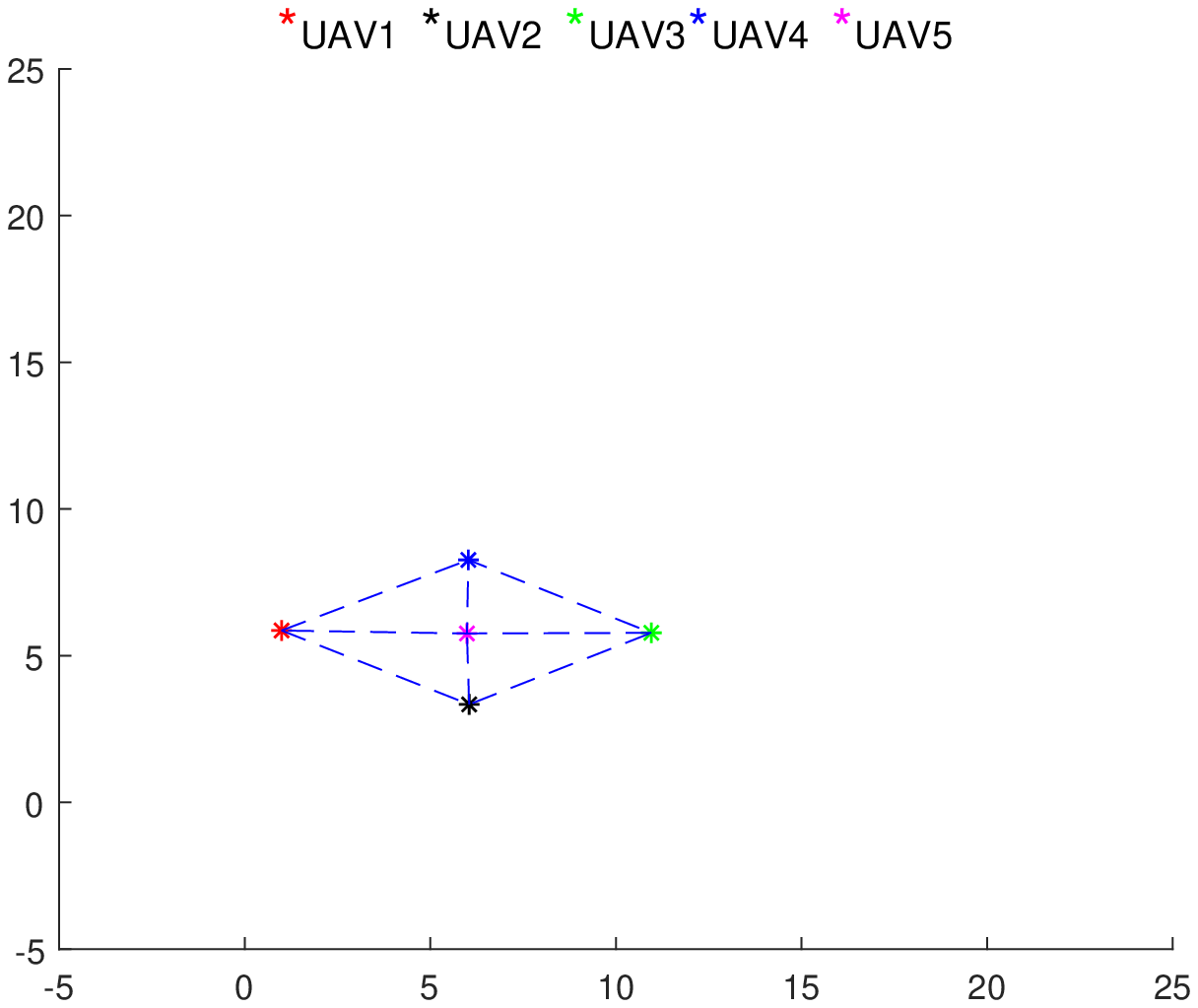}}
\quad
\subfloat[Simulation II: $\alpha=5$, $\beta=0.3$]{\label{Fig.4(b)}\includegraphics[width=0.3\textwidth]{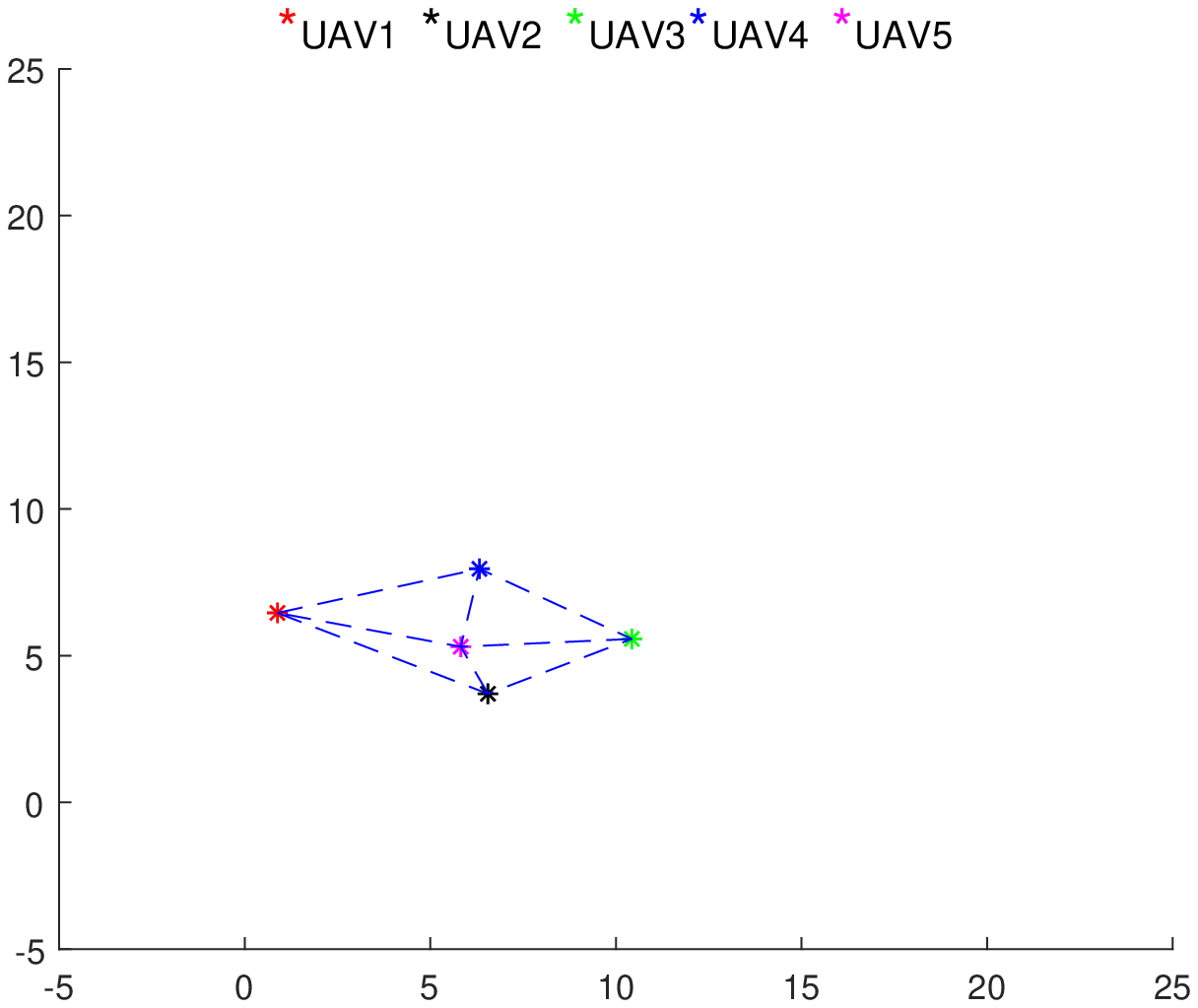}}
\quad
\subfloat[Simulation III: $\alpha=853$, $\beta=0.7$]{\label{Fig.4(c)}\includegraphics[width=0.3\textwidth]{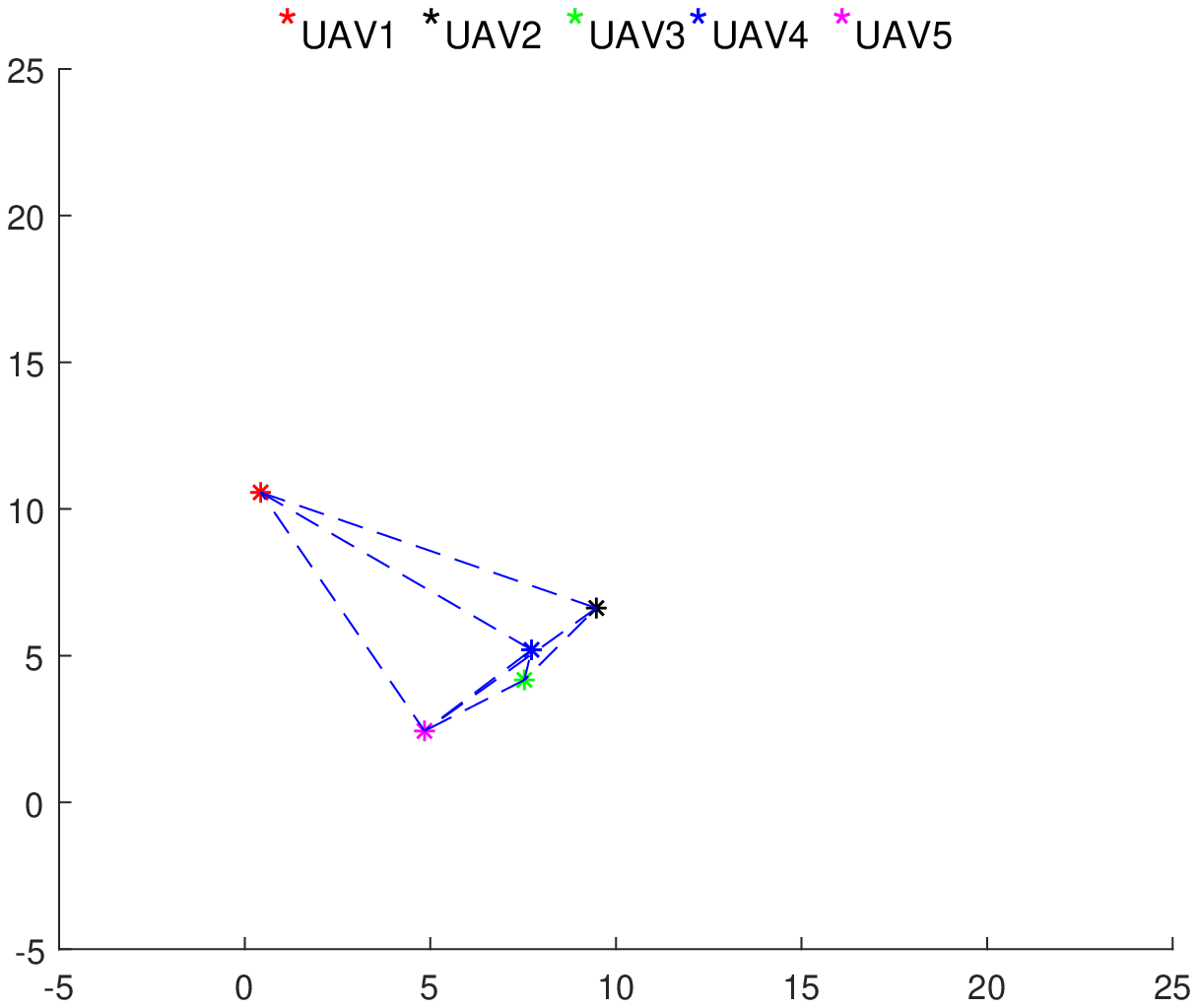}}\\
\caption{The final formation shape of the multi-agent system. The position of each agent is indicated by $*$.}
\label{Fig.4}
\end{figure*}

\begin{figure*}[htp!]
\centering
\setcounter{figure}{4}
\subfloat[]{\label{Fig.5(a)}\includegraphics[width=0.3\textwidth]{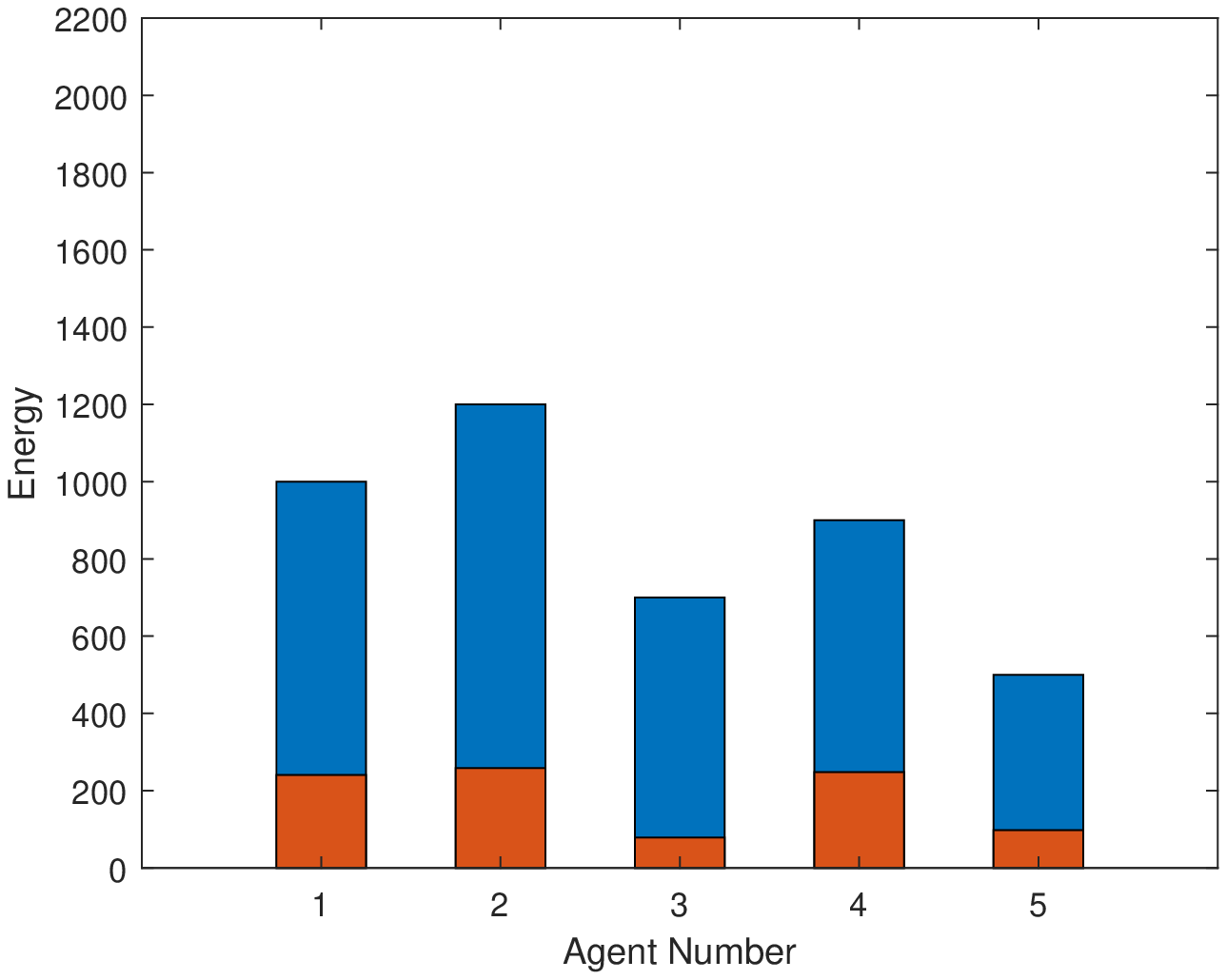}}
\quad
\subfloat[]{\label{Fig.5(b)}\includegraphics[width=0.3\textwidth]{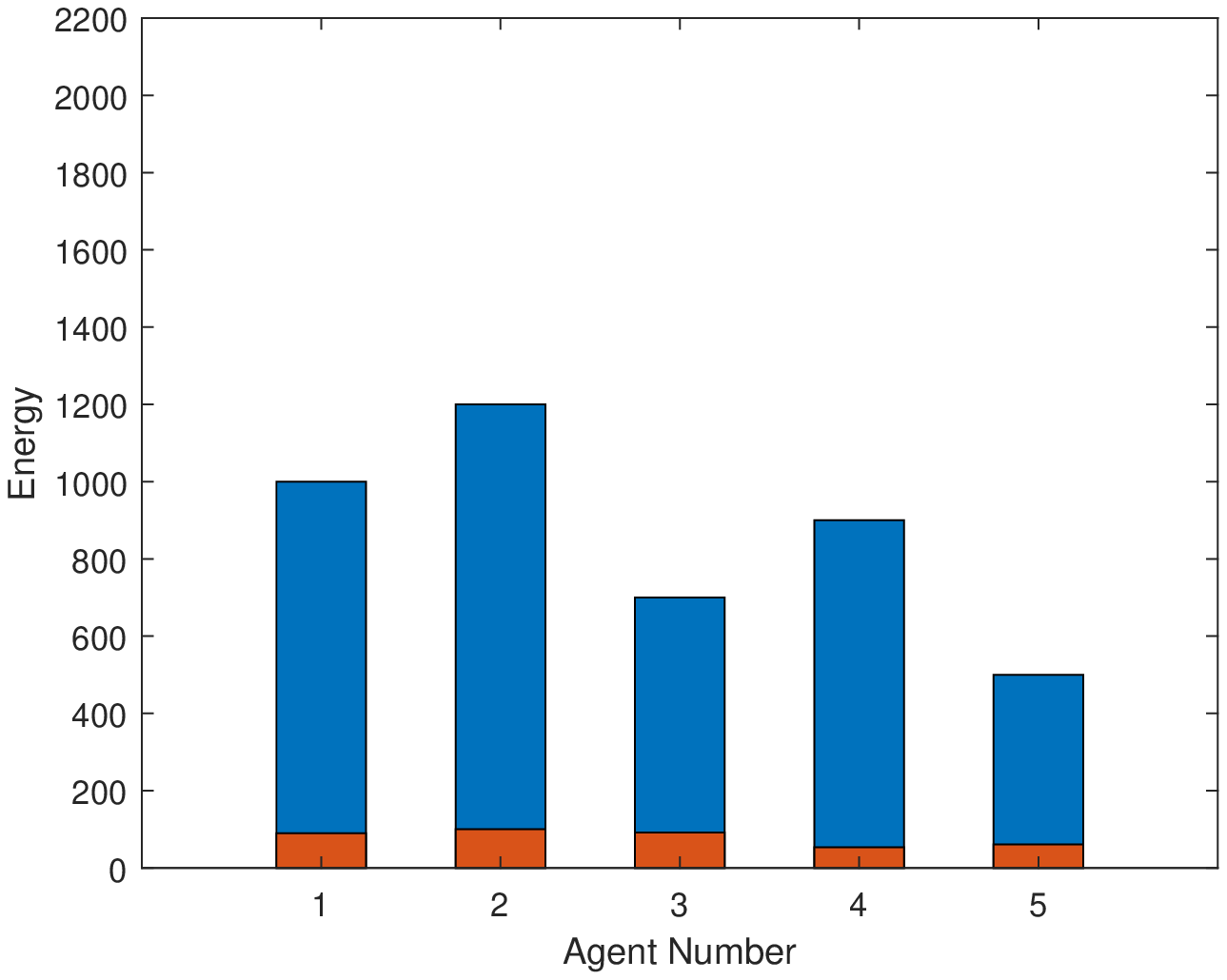}}
\quad
\subfloat[]{\label{Fig.5(c)}\includegraphics[width=0.3\textwidth]{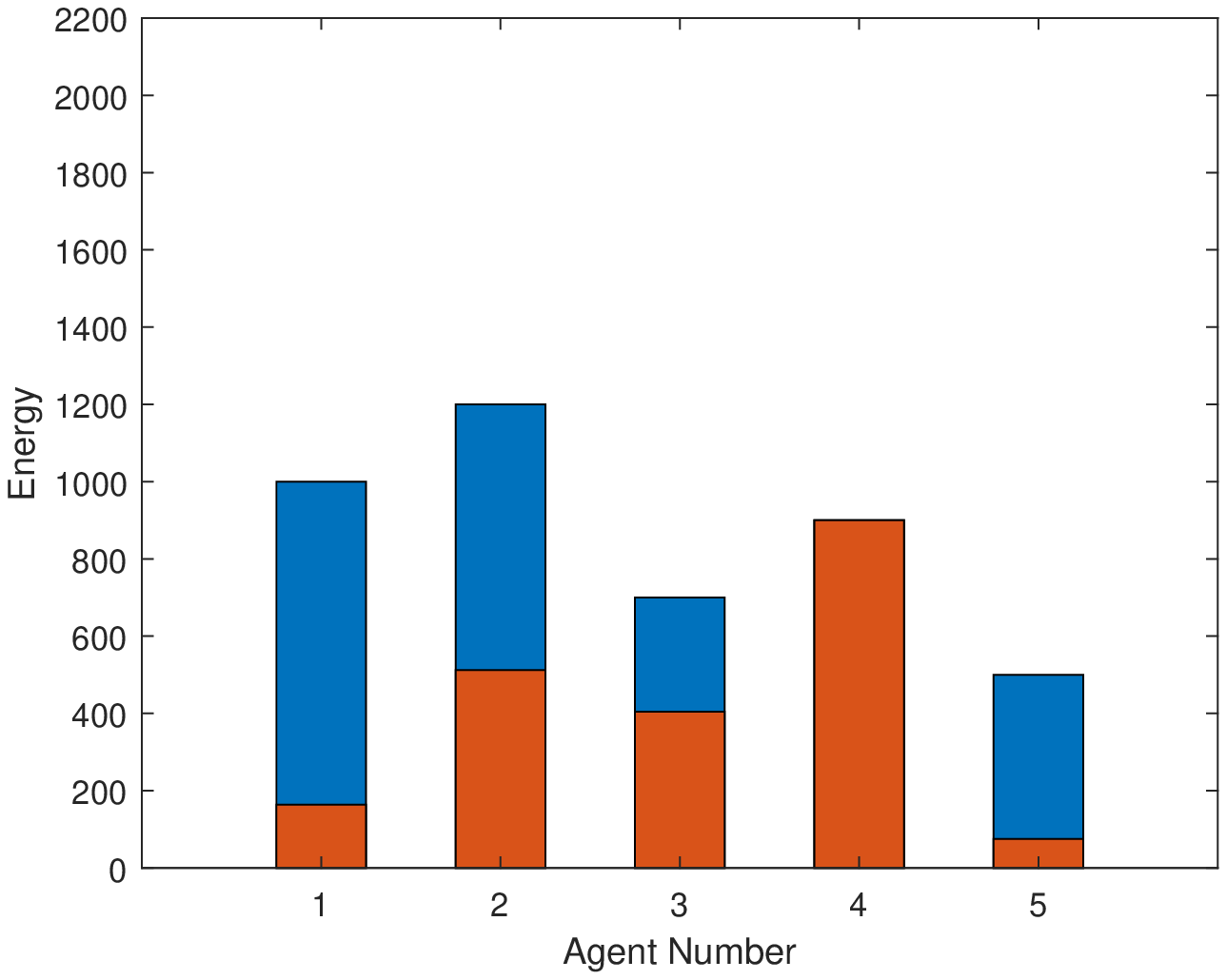}}\\
\caption{The energy consumption of the multi-agent system. The red color indicates the energy consumed, whilst the blue color indicates the remaining energy level.}
\label{Fig.5}
\end{figure*}

\begin{table}[h]
\setlength{\tabcolsep}{1.45mm}
\centering
\caption{Values of the parameters}
\begin{tabular}{c|c|c|c|c}
\hline
Simulation & Value of & Value of &Value of &Formation \\
  Number   & $\alpha$ & $\sigma$ & $\beta$ &Time $t_f$(s)   \\
\hline
\uppercase\expandafter{\romannumeral1} & 450& 1.3 & 0.2 &0.49\\
\uppercase\expandafter{\romannumeral2} & 5& 1.3& 0.3 &N/A\\
\uppercase\expandafter{\romannumeral3} &  853& 1.3 &0.7 &N/A\\
\hline
\end{tabular}
\label{tab1}
\end{table}

Table~\ref{tab1} shows three sets of the values of $\alpha$, $\sigma$, and $\beta$. Only the first set satisfies the energy and time constraints, i.e., Eqs.~(\ref{12}) and (\ref{13}), simultaneously. The second set violates the termination time constraint (\ref{12}), while the third set violates the energy constraint \eqref{13}. Fig.~\ref{Fig.4} depicts the final formation shape of the multi-agent system in each case. Fig.~\ref{Fig.5} shows the energy consumption of the agents during the formation task. It can be observed that in the first case, the formation is achieved and the energy is not exhausted for each agent; in the second case, the formation task is not accomplished by the end of the termination time $T=3s$; in the third case, the energy of agent~$4$ is exhausted before the formation mission is accomplished.

\section{Conclusions}
\label{sect:concl}
This paper presents a globally optimal distributed formation control algorithm and a comprehensive analysis of the roles of energy levels, termination time, control parameters, as well as the network topology on achieving energy and time constrained formation control. Two lower bounds on the required initial energy levels and on the achievable termination time  are explicitly given, which help answer the question whether a distributed formation control problem is feasible under prescribed hard constraints on the termination time and energy expenditure. Additionally, several monotonicity properties in relation to the control parameters, in particular, the achievable termination time and the required initial energy with respect to those control parameters are derived. These properties can be properly exploited to facilitate the formation control design. The formulation of this paper provides a solution to LQR-based formation control under constraints of both termination time and energy. The future topic can be directed to nonlinear agent dynamics and directed network topologies.  




\bibliographystyle{IEEEtran}
\bibliography{refs}

\begin{thebibliography}{10}
\providecommand{\url}[1]{#1}
\csname url@samestyle\endcsname
\providecommand{\newblock}{\relax}
\providecommand{\bibinfo}[2]{#2}
\providecommand{\BIBentrySTDinterwordspacing}{\spaceskip=0pt\relax}
\providecommand{\BIBentryALTinterwordstretchfactor}{4}
\providecommand{\BIBentryALTinterwordspacing}{\spaceskip=\fontdimen2\font plus
\BIBentryALTinterwordstretchfactor\fontdimen3\font minus
  \fontdimen4\font\relax}
\providecommand{\BIBforeignlanguage}[2]{{%
\expandafter\ifx\csname l@#1\endcsname\relax
\typeout{** WARNING: IEEEtran.bst: No hyphenation pattern has been}%
\typeout{** loaded for the language `#1'. Using the pattern for}%
\typeout{** the default language instead.}%
\else
\language=\csname l@#1\endcsname
\fi
#2}}
\providecommand{\BIBdecl}{\relax}
\BIBdecl

\bibitem{oh2015survey}
K.-K. Oh, M.-C. Park, and H.-S. Ahn, ``A survey of multi-agent formation
  control,'' \emph{Automatica}, vol.~53, pp. 424--440, 2015.

\bibitem{su2009flocking}
H.~Su, X.~Wang, and Z.~Lin, ``Flocking of multi-agents with a virtual leader,''
  \emph{IEEE Transactions on Automatic Control}, vol.~54, no.~2, pp. 293--307,
  2009.

\bibitem{chen2017connection}
F.~Chen and W.~Ren, ``A connection between dynamic region-following formation
  control and distributed average tracking,'' \emph{IEEE Transactions on
  Cybernetics}, vol.~48, no.~6, pp. 1760--1772, 2017.

\bibitem{beard2001coordination}
R.~W. Beard, J.~Lawton, and F.~Y. Hadaegh, ``A coordination architecture for
  spacecraft formation control,'' \emph{IEEE Transactions on Control Systems
  Technology}, vol.~9, no.~6, pp. 777--790, 2001.

\bibitem{balch1998behavior}
T.~Balch and R.~C. Arkin, ``Behavior-based formation control for multirobot
  teams,'' \emph{IEEE Transactions on Robotics and Automation}, vol.~14, no.~6,
  pp. 926--939, 1998.

\bibitem{lin2005necessary}
Z.~Lin, B.~Francis, and M.~Maggiore, ``Necessary and sufficient graphical
  conditions for formation control of unicycles,'' \emph{IEEE Transactions on
  Automatic Control}, vol.~50, no.~1, pp. 121--127, 2005.

\bibitem{weimerskirch2001energy}
H.~Weimerskirch, J.~Martin, Y.~Clerquin, P.~Alexandre, and S.~Jiraskova,
  ``Energy saving in flight formation,'' \emph{Nature}, vol. 413, no. 6857, pp.
  697--698, 2001.

\bibitem{derenick2011energy}
J.~Derenick, N.~Michael, and V.~Kumar, ``Energy-aware coverage control with
  docking for robot teams,'' in \emph{2011 IEEE/RSJ International Conference on
  Intelligent Robots and Systems}.\hskip 1em plus 0.5em minus 0.4em\relax IEEE,
  2011, pp. 3667--3672.

\bibitem{papakostas2018energy}
D.~Papakostas, S.~Eshghi, D.~Katsaros, and L.~Tassiulas, ``Energy-aware
  backbone formation in military multilayer ad hoc networks,'' \emph{Ad Hoc
  Networks}, vol.~81, pp. 17--44, 2018.

\bibitem{sardellitti2011optimal}
S.~Sardellitti, S.~Barbarossa, and A.~Swami, ``Optimal topology control and
  power allocation for minimum energy consumption in consensus networks,''
  \emph{IEEE Transactions on Signal Processing}, vol.~60, no.~1, pp. 383--399,
  2011.

\bibitem{babazadeh2018cooperative}
R.~Babazadeh and R.~Selmic, ``Cooperative distance-based leader-following
  formation control using sdre for multi-agents with energy constraints,'' in
  \emph{2018 IEEE Conference on Decision and Control (CDC)}.\hskip 1em plus
  0.5em minus 0.4em\relax IEEE, 2018, pp. 5008--5014.

\bibitem{babazadeh2018anoptimal}
------, ``An optimal displacement-based leader-follower formation control for
  multi-agent systems with energy consumption constraints,'' in \emph{2018 26th
  Mediterranean Conference on Control and Automation (MED)}.\hskip 1em plus
  0.5em minus 0.4em\relax IEEE, 2018, pp. 179--184.

\bibitem{zhang2018consensus}
H.~Zhang and X.~Hu, ``Consensus control for linear systems with optimal energy
  cost,'' \emph{Automatica}, vol.~93, pp. 83--91, 2018.

\bibitem{moarref2014optimal}
M.~Moarref and L.~Rodrigues, ``An optimal control approach to decentralized
  energy-efficient coverage problems,'' \emph{IFAC Proceedings Volumes},
  vol.~47, no.~3, pp. 6038--6043, 2014.

\bibitem{mei2015distributed}
J.~Mei, W.~Ren, and J.~Chen, ``Distributed consensus of second-order
  multi-agent systems with heterogeneous unknown inertias and control gains
  under a directed graph,'' \emph{IEEE Transactions on Automatic Control},
  vol.~61, no.~8, pp. 2019--2034, 2015.

\bibitem{xiang2019advances}
L.~Xiang, F.~Chen, W.~Ren, and G.~Chen, ``Advances in network
  controllability,'' \emph{IEEE Circuits and Systems Magazine}, vol.~19, no.~2,
  pp. 8--32, 2019.

\bibitem{demirel2017trade}
B.~Demirel, A.~S. Leong, V.~Gupta, and D.~E. Quevedo, ``Trade-offs in
  stochastic event-triggered control,'' \emph{arXiv preprint arXiv:1708.02756},
  2017.

\bibitem{varma2019energy}
V.~S. Varma, A.~M. de~Oliveira, R.~Postoyan, I.-C. Morarescu, and J.~Daafouz,
  ``Energy-efficient time-triggered communication policies for wireless
  networked control systems,'' \emph{IEEE Transactions on Automatic Control},
  2019.

\bibitem{niu2017numerical}
J.-Q. Niu, D.~Zhou, T.-H. Liu, and X.-F. Liang, ``Numerical simulation of
  aerodynamic performance of a couple multiple units high-speed train,''
  \emph{Vehicle System Dynamics}, vol.~55, no.~5, pp. 681--703, 2017.

\bibitem{chu2014numerical}
C.-R. Chu, S.-Y. Chien, C.-Y. Wang, and T.-R. Wu, ``Numerical simulation of two
  trains intersecting in a tunnel,'' \emph{Tunnelling and Underground Space
  Technology}, vol.~42, pp. 161--174, 2014.

\bibitem{cao2009optimal}
Y.~Cao and W.~Ren, ``Optimal linear-consensus algorithms: An lqr perspective,''
  \emph{IEEE Transactions on Systems, Man, and Cybernetics, Part B
  (Cybernetics)}, vol.~40, no.~3, pp. 819--830, 2009.

\bibitem{di2012rendezvous}
S.~Di~Cairano, C.~A. Pascucci, and A.~Bemporad, ``The rendezvous dynamics under
  linear quadratic optimal control,'' in \emph{2012 IEEE 51st IEEE Conference
  on Decision and Control (CDC)}.\hskip 1em plus 0.5em minus 0.4em\relax IEEE,
  2012, pp. 6554--6559.

\bibitem{chen2019minimum}
F.~Chen and J.~Chen, ``Minimum-energy distributed consensus control of
  multi-agent systems: A network approximation approach,'' \emph{IEEE
  Transactions on Automatic Control}, 2019.

\bibitem{jia2020distributed}
C.~Jia, F.~Chen, L.~Xiang, and W.~Lan, ``Distributed optimal formation control
  with hard constraints on energy and time,'' in \emph{2020 IEEE International
  Conference on Control and Automation (ICCA)}.\hskip 1em plus 0.5em minus
  0.4em\relax IEEE, accepted, 2020.

\end{thebibliography}

\end{document}